\documentclass[10pt,pre, aps, amsmath, amssymb, twocolumn, superscriptaddress,nofootinbib]{revtex4-2}

\usepackage[T1]{fontenc}
\usepackage[english]{babel}
\usepackage{varwidth}
\usepackage{amsmath}
\usepackage{amsthm}
\usepackage{amssymb}
\usepackage{graphicx}
\usepackage{hyperref}
\usepackage{orcidlink} 
\usepackage{comment}
\usepackage{colortbl}
\usepackage{multirow}
\makeatletter

\pdfpageheight\paperheight
\pdfpagewidth\paperwidth

\providecommand{\tabularnewline}{\\}
\newenvironment{cellvarwidth}[1][t]
    {\begin{varwidth}[#1]{\linewidth}}
    {\@finalstrut\@arstrutbox\end{varwidth}}

\theoremstyle{plain}
\newtheorem{thm}{\protect\theoremname}
\theoremstyle{definition}
\newtheorem{defn}{\protect\definitionname}
\theoremstyle{plain}
\newtheorem*{lem*}{\protect\lemmaname}

\usepackage{times}

\renewcommand\[{\begin{equation}}
\renewcommand\]{\end{equation}}

\DeclareMathOperator{\tr}{tr}

\newcommand{\1}{\leavevmode{\mathrm{1\ifmmode\mkern  -4.8mu\else\kern -.3em\fi I}}}

\makeatother

\providecommand{\definitionname}{Definition}
\providecommand{\lemmaname}{Lemma}
\providecommand{\theoremname}{Theorem}

\begin{document}

\title{Integrability and Chaos via fractal analysis of Spectral Form Factors:\\
 Gaussian approximations and exact results}

\author{Lorenzo Campos Venuti\orcidlink{0000-0002-0217-6101}}
\email[Corresponding author: ]{lorenzo.camposvenuti@unina.it}
\affiliation{Dipartimento di Fisica `Ettore Pancini', Universit\`a degli Studi di Napoli Federico II, Via Cintia 80126,  Napoli, Italy}
\affiliation{Department of Physics and Astronomy, University of Southern California, Los Angeles, USA}
\author{Jovan Odavi\'c\orcidlink{0000-0002-0217-61013-2729-8284}}
\affiliation{Dipartimento di Fisica `Ettore Pancini', Universit\`a degli Studi di Napoli Federico II, Via Cintia 80126,  Napoli, Italy}
\affiliation{Istituto Nazionale di Fisica Nucleare (INFN), Sezione di Napoli, Italy}
\author{Alioscia Hamma\orcidlink{0000-0003-0662-719X}}
\affiliation{Dipartimento di Fisica `Ettore Pancini', Universit\`a degli Studi di Napoli Federico II, Via Cintia 80126,  Napoli, Italy}
\affiliation{Istituto Nazionale di Fisica Nucleare (INFN), Sezione di Napoli, Italy}
\affiliation{Scuola Superiore Meridionale, Largo S. Marcellino 10, 80138 Napoli, Italy}


\begin{abstract}
It is well known that the spectral form factor (SFF) of a possibly degenerate many-body Hamiltonian can be identified with a planar random walk taking steps of unequal length. In this paper we push this identification further and propose to study the chaotic content of a Hamiltonian $H$ via its associated random walk seen as a fractal, using the tools of fractal geometry. In particular we conjecture that for chaotic Hamiltonians the Hausdorff dimension of the frontier of the corresponding random walk approaches the universal value $d_F=4/3$ -- the same value obtained when the random walk describes a Wiener process.
Our numerical simulations for non-integrable models confirm this expectation while for quasi-free integrable models we obtain a value $d_F = 1$. Additionally, we
numerically show that ``Bethe Ansatz walkers'' fall into a category
similar to the non-integrable walkers.
To motivate this conjecture we consider many-body Hamiltonians with degenerate but rationally independent eigenvalues. We prove that if the degeneracies satisfy certain Lyapunov conditions, the random walk becomes a Wiener process, $d_F=4/3$, and the distribution of the SFF becomes Gaussian. This is the familiar Gaussian approximation for the SFF which we show to be violated at very low temperature. We also compute the moments of the SFF exactly under milder hypotheses thus solving the classical problem of determining the moments of a random walker taking steps of unequal lengths. Finally, we consider quasi-free Fermionic models with possibly degenerate but rationally independent one-particle spectra. 
We show that in this case the distribution of the SFF becomes log-Normal and also give the exact form of the moments under milder hypotheses. 
\end{abstract}


\maketitle


\section{Introduction}
Unitarity of quantum mechanics implies that the solutions of the Schr\"{o}dinger equation with nearby initial conditions remain close throughout the evolution, meaning that chaos in quantum mechanics cannot be defined via the sensitivity to the initial conditions -- the familiar \emph{butterfly effect} of chaotic classical mechanics~\cite{Goldstein_Safko_2002}. Instead, in order to characterize and understand quantum chaos, one has to resort to other quantities and ideas, such as the energy level statistics \cite{berry_level_1977,bohigas_characterization_1984} or the out-of-time-ordered correlators (OTOCs) \cite{maldacena_bound_2016,cotler_out--time-order_2018}. Another object used to probe quantum chaos is the spectral form factor (SFF)
 $S(t)$~\cite{Haake_2001}. The SFF is a sort of simplified version of a quantum expectation value and depends only on the Hamiltonian $H$ under consideration and a possibly not normalized quantum state $\rho$. It is precisely the SFF that will be at the center of our investigations. It has the form (we use $\hbar=1$ throughout the text)
\begin{equation}
S(t)=\left|\chi(t)\right|^{2},\,\,\chi(t)=\tr\left(\rho e^{-itH}\right)\,.\label{eq:start}
\end{equation}
The number of interesting physical problems --either quantum or classical-- that can be written, often exactly, in this form,  is  quite amazing. Lagrange showed that the study of the motion of perihelia is connected to the argument of an expression like $\chi(t)$ with $\rho=\1$ \citep{lagrange_theorie_1781}.
For a pure quantum state $\rho=|\psi\rangle\langle\psi|$,
$S(t)$ is known as Loschmidt echo or survival probability and is
connected to the theory of Fermi-edge singularity in metals \citep{schotte_tomonagas_1969},
the statistics of the work done after a quench \citep{silva_statistics_2008},
the approach to thermodynamic equilibrium \citep{campos_venuti_unitary_2010,campos_venuti_equilibration_2013,campos_venuti_universal_2014}
and (especially for $\rho =\1$) quantum chaos~ \citep{berry_semiclassical_1985,brezin_spectral_1997,prosen_time_1998,prosen_ergodic_1999,kunz_probability_1999,cotler_chaos_2017,roy_random_2020,suntajs_spectral_2021,roy_spectral_2022,
altland_statistics_2025}.
More recently, the SFF has found applications in quantum gravity \citep{cotler_black_2017,CotlerErratum2018,caceres_spectral_2022,choi_supersymmetric_2023, delCampo_Molina-Vilaplana_Sonner_2017} and non-unitary dynamics~\cite{Xu_Chenu_Prosen_delCampo_2021,Cornelius_Xu_Saxena_Chenu_delCampo_2022,MatsoukasRoubeas_Beau_Santos_delCampo_2023}. Using the Sachdev-Ye-Kitaev (SYK) Hamiltonian~\cite{sachdev1993gapless, KITP2015} as a model for a black hole, the late-time behavior
of the SFF probes the discreteness of the black hole spectrum. Three
regimes have been identified for the SFF. A short-time regime (the
``slope'') where $S(t)$ is observed to be self-averaging, a long
time ``plateau'' where $S(t)$ is approximately constant but subject
to large fluctuations, and an intermediate regime connecting the
two (the ``ramp''). It was recognized already in \citep{campos_venuti_unitary_2010,liu_spectral_2018}
that a cumulant expansion implies $S(t)\simeq e^{-\left(\left\langle H^{2}\right\rangle -\left\langle H\right\rangle ^{2}\right)t^{2}}$
for short times, which, due to the self-averaging property of the
SYK moments $\left\langle H^{2p}\right\rangle $, implies the observed
universal short-time regime. For larger time the cumulant expansion
breaks down and $S(t)$ start oscillating erratically. It is precisely this challenging regime upon which we shall put more emphasis in this work. 
It has been
known for quite some time that the value of the plateau is related
to a particular Gaussian approximation \cite{prange_spectral_1997,suntajs_spectral_2021} (we will come back to this point below). 

Yet another way to look at Eq.~(\ref{eq:start})
is as the position of a random walker on the complex plane. This analogy was
first utilized by Rayleigh \citep{rayleigh_problem_1905,pearson_problem_1905}
and the celebrated Rayleigh distribution is precisely the distribution
of the distance of the walker from the origin $\left|\chi(t)\right|$,
under the above mentioned Gaussian approximation. 

In this paper, we further exploit the random walk analogy. The entire
path of the walker is seen as a random process of discrete time and can be studied with the tools of fractal geometry~\cite{mandelbrot_fractal_1982}. We propose to study the chaotic content of the many-body Hamiltonian $H$ through features of its random walk, seen as a fractal, in particular the Hausdorff dimension of its frontier.  We conjecture that for chaotic Hamiltonians the dimension of the frontier approaches the value $4/3$ 
in the thermodynamic limit, 
which is the value famously obtained in \cite{lawler_dimension_2001} when the walker describes a Wiener process.  This conjecture is similar in spirit to conjecture A of Ref.~\cite{aurich_temporal_1999} formulated for one-body systems, which amounts to say that the (renormalized) distribution of the distance of the walker $\left|\chi(t)\right|$ has  the universal Rayleigh form for quantum systems whose classical dynamics is chaotic. Indeed, both our conjecture and conjecture A of Ref.~\cite{aurich_temporal_1999} are implied by a Wiener walker. 

Our numerics
for non-integrable models show that the dimension of the
frontier of the random walk $d_F$ is close to $4/3$. 
Instead,   for quasi-free systems we obtain a value close to one. An illustration of the fractals that can be obtained in the two different phases is shown  in Fig.~\ref{fig:mainplot} whereas our results for the dimensions of the associated frontiers are summarized in Table \ref{tab:Expectation-of-different}. Our results for Bethe-Ansatz integrable models are consistent with a value of $d_F < 4/3$ but are at the moment inconclusive. 
A potential benefit of this approach, is that the number of steps performed by the walker is given by the exponentially large Hilbert's space dimension, and as such, relatively small sizes are sufficient to obtain manifestly fractal behavior. The main shortcoming of this approach is the complexity of determining the fractal frontier, which at the moment cannot be fully automated (see text for details).

To motivate our conjecture we consider many-body systems with possibly degenerate spectra, as is the case for most physical models, chaotic or not. We show that, if the energy eigenvalues are linearly independent over the rationals, and certain Lyapunov conditions for the degeneracies are satisfied, the walker describes a Wiener process, the probability distribution of the SFF becomes Gaussian and the dimension of the frontier is $4/3$. This is the familiar ``Gaussian approximation'' well known in the context of the spectral form factor \cite{prange_spectral_1997,suntajs_spectral_2021} generalized to the case of degenerate spectra. In particular we show that at sufficiently low temperature and for $\rho = |\psi \rangle \langle \psi |$ 
for a small quench at a critical point, the Lyapunov condition is never satisfied and the Gaussian approximation breaks down. 

We also
give the exact expression for the moments of the SFF beyond the
Gaussian approximation as  a closed-form recursion and in full generality accounting for possible degeneracies in the spectrum. 
These exact results extend, in the plateau region, the beyond-Gaussian approximation obtained in the ramp-plateau regime of \cite{altland_statistics_2025} (see also \cite{Haake_2001,flack_statistics_2020,suntajs_quantum_2020,chan_spectral_2021,winer_hydrodynamic_2022,legramandi_moments_2024,kumar_leading_2025}), and can serve as a check for further approximations on the time dependence of the SFF. 

Finally, we consider  quasi-free systems with Fermionic rationally independent one-particle spectra. 
While, in the generic non-integrable case,
$\left|\chi(t)\right|$ is a sum of random variables, for quasi-free
models, $\left|\chi(t)\right|$ is rather a product of 
variables, and in case of rationally independent spectra an analogous central limit theorem applies for $\log \left|\chi(t)\right|$ and the distribution of the SFF becomes log-normal. We also give the exact form of the moments of the SFF in this case without resorting to the Gaussian approximation.

The paper is organized as follows. In Section II we connect time and ensemble averages (Theorem 1) and provide the standard theorems relating $\chi(t)$ to a Brownian motion. In Section III we give the exact moments of the SFF of any order, without resorting to any approximation. In Section IV we consider integrable quadratic Fermions and provide the limiting distribution of the SFF in this case, which turns out to be a LogNormal (Theorem 5), and also give the exact form of the moments without any approximation (Theorem 6). In Section V we provide numerical results based on the XXZ chain with next nearest neighbor. We extract the Hausdorff dimension of the frontier related to quadratic, Bethe-Ansatz integrable, and non-integrable cases. The results are summarized in Table \ref{tab:Expectation-of-different}.

\begin{table}
\begin{centering}
\begin{tabular}{|c|c|c|c|}
\hline 
 & Quadratic & \begin{cellvarwidth}[t]
\centering
Bethe-\\
Ansatz
\end{cellvarwidth} & \begin{cellvarwidth}[t]
\centering
Non-\\
Integrable
\end{cellvarwidth}\tabularnewline
\hline 
\hline 
\begin{cellvarwidth}[t]
\centering
Level spacing\\
distribution
\end{cellvarwidth} & Poisson & Poisson & Wigner-Dyson \tabularnewline
\hline 
\begin{cellvarwidth}[t]
\centering
Normalized $\left|\chi(t)\right|^{2}$\\
 distribution
\end{cellvarwidth} & \textsf{$\mathsf{LogNormal}$} & $\mathsf{Exp}(1)$ & $\mathsf{Exp}(1)$\tabularnewline
\hline 
\begin{cellvarwidth}[t]
\centering
Dimension\\
 of frontier
\end{cellvarwidth} & $1.01\pm 0.04$ & $1.24 \pm0.08$ & $1.32\pm 0.08$\tabularnewline
\hline 
\end{tabular}
\par\end{centering}
\caption{Expectation of different metrics of chaos for increasingly more chaotic
models (from left to right).  See e.g.~\cite{Pozsgay_2021} for the first row. \label{tab:Expectation-of-different}}
\end{table}


\section{The Fractal Geometry of Spectral Statistics}

\begin{figure*}
\begin{centering}
\hspace{14mm}
\begin{minipage}{80mm}
\includegraphics[width=80mm]{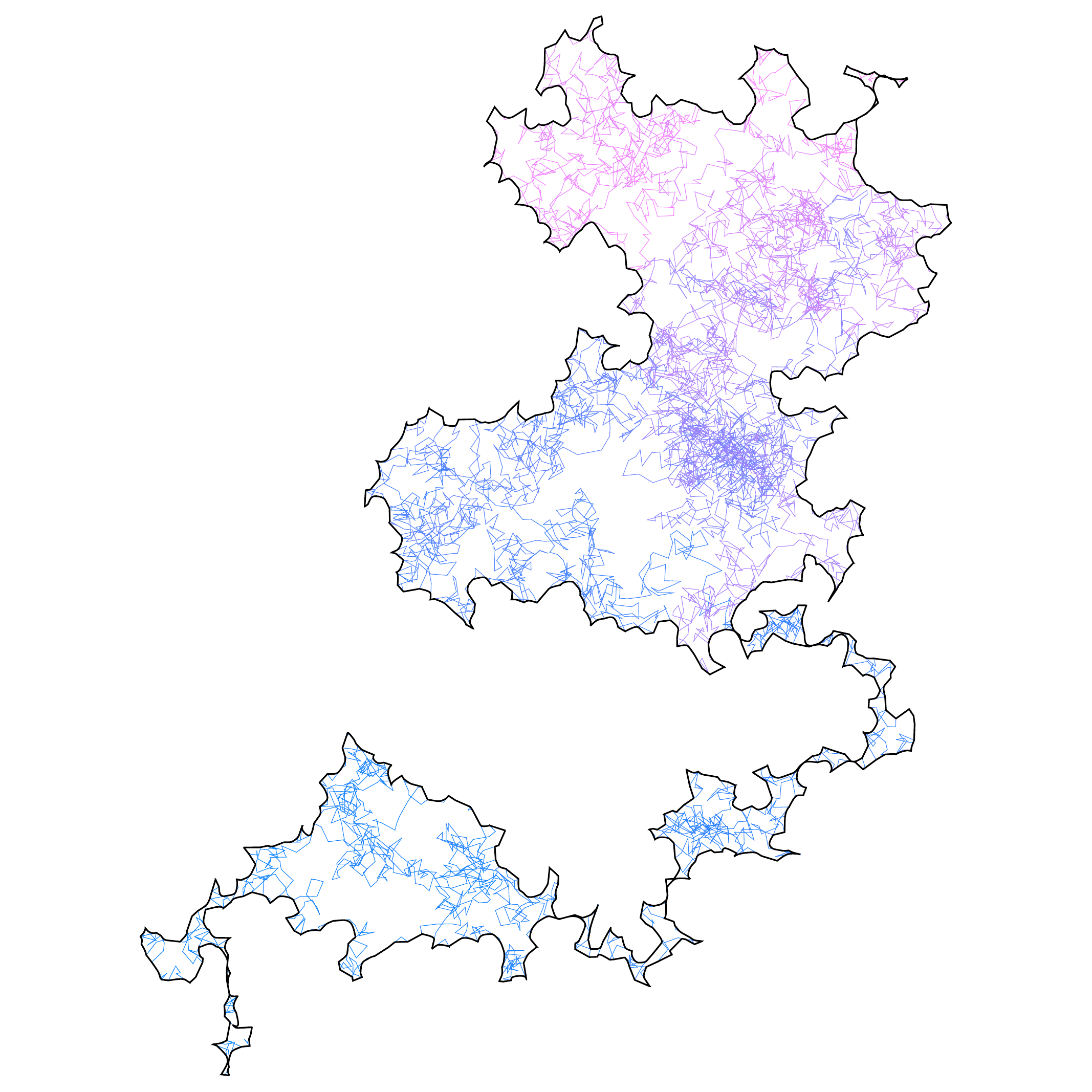}
\end{minipage}
\vspace{11pt}
\hspace{7mm}
\begin{minipage}{70mm}
\includegraphics[width=55mm]{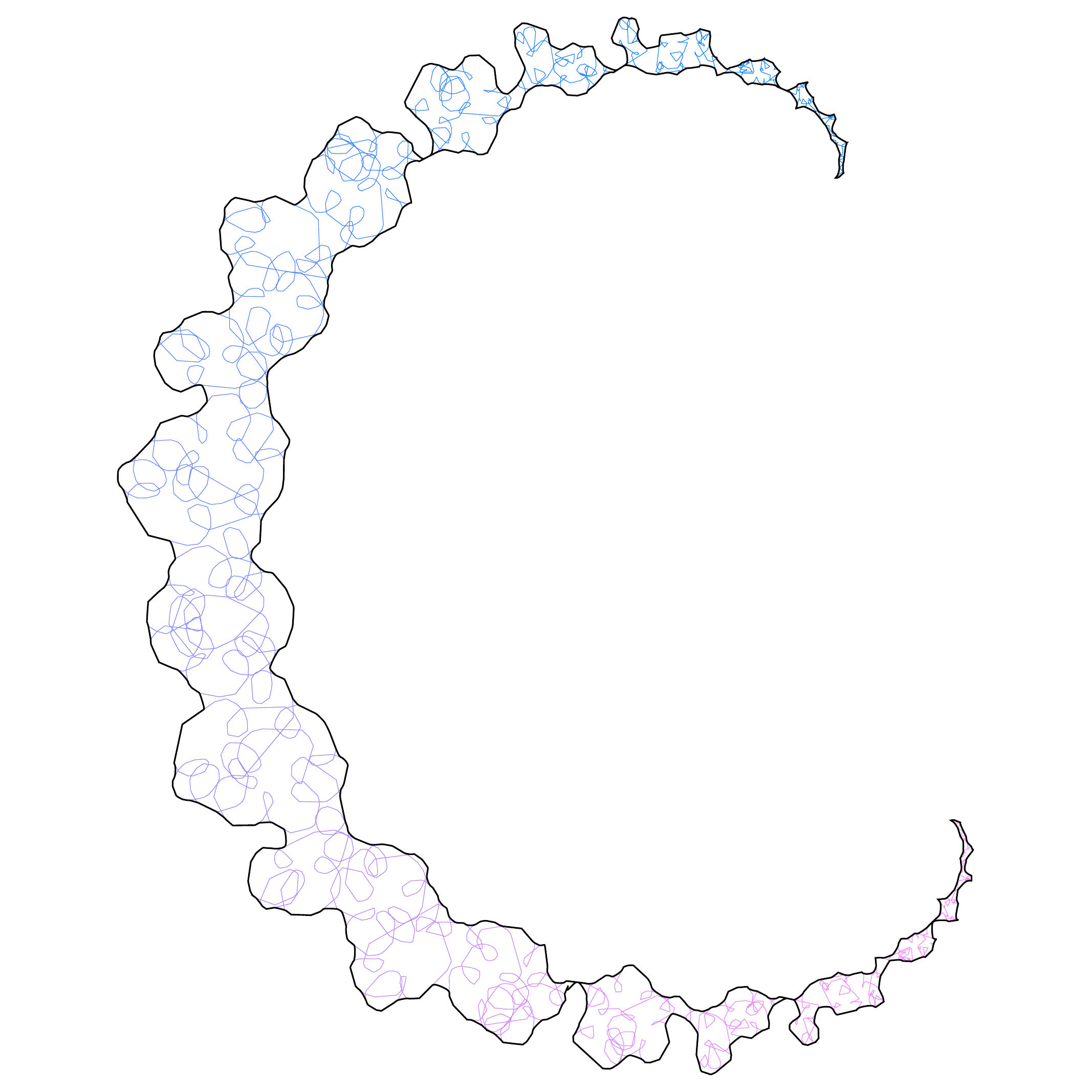}
\end{minipage}
\vspace{-7mm}
\end{centering}
\caption{Fractals and their frontiers, in black, corresponding to physical models: non-integrable (left) vs integrable (right). Color from blue (start) to pink (end) corresponds to increasing time-steps of the random walks. The frontier is essentially the boundary of the fractal without the inner islands, see text for details. The non-integrable Hamiltonian is the XXZ model with next nearest neighbor interactions (see Eq.~(\ref{XXZNNNHamiltonian})), with $(\Delta,\alpha) = (0.4,0.5)$. The integrable Hamiltonian is the XY model with parameters $(h,\gamma)=(0.2,0.3)$; see Appendix~\ref{sec:appfractral}. 
\label{fig:mainplot}}
\end{figure*}

For each ``system size'' $L$, we assume the Hamiltonian $H$ lives in a Hilbert space
of dimension $D$ e.g.~for qubits $D=2^L$. We allow the Hamiltonian to have degenerate eigenvalues. Indeed, general physical models, even non-integrable, chaotic ones, possess several symmetries and degeneracies. Using the spectral resolution of the Hamiltonian,
$H=\sum_{j=1}^{N_{B}}E_{j}\Pi_{j}$ ($E_{j}$ eigenenergies, $\Pi_{j}$
eigenprojectors, $N_{B}$ the number of blocks) we have, with $d_{j}=\tr\left(\rho\Pi_{j}\right)>0$,
\begin{equation}
\chi(t)=\sum_{j=1}^{N_{B}}d_{j}e^{-itE_{j}}.\label{eq:chit}
\end{equation}

We now show how to associate to any pair $\left(H,\rho\right)$ a
random walk on the plane. In this way, we encode information about
the pair $\left(H,\rho\right)$ into a single function, more precisely
into a stochastic process. This analogy has already been used in \citep{prange_spectral_1997}
but we are going to follow its consequences further. For each system
size $L$ we diagonalize the Hamiltonian $H$ and obtain the sets
$\left(\{ E_{j}\}^\uparrow ,\left\{ d_{j}\right\} \right)$ of ascending eigenvalues
$E_{1}<E_{2}<\cdots<E_{N_{B}}$. Time $t$ is a random variable acted
upon by the following expectation value, 
\[
\overline{f(t)}:=\lim_{T\to\infty}\frac{1}{T}\int_{0}^{T}dtf(t)\,,
\]
in other words, in the spirit of ergodic theory and statistical mechanics,
time is subject to infinite time average. We thus obtain an ensemble
of walks on the plane, indexed by a random variable $t$, i.e.,~a
random walk. At step $n$ the walker is at position $Z_{n}(t)=\sum_{j=1}^{n}d_{j}e^{-itE_{j}}$
on the complex plane ($Z_{0}(t)=0$), $n=1,2,\ldots,N_{B}$. At step
$n+1$ the walker rotates clockwise by an angle $tE_{n+1}$ and takes
a step of length $d_{n+1}$. For each $L$ we only have a walker with
a finite (albeit exponentially large) number of steps, and we are
naturally interested in what happens in the thermodynamic limit when
$L\to\infty$. Note that the path of the walker is strongly dependent
on the ordering of the energies. Our natural ordering is consistent
with the analogous choice in random matrix theory (RMT). 

Unlike
the more familiar random walks defined on a lattice, this walker can
occupy any point on the complex plane $\simeq\mathbb{R}^{2}$.\footnote{We move from $\mathbb{C}$ to $\mathbb{R}^{2}$ whichever is more
convenient.} A useful quantity to describe a random walk is the probability that
the walker is at position $\boldsymbol{z}=(z_{1},z_{2})$ after $n$
steps. Since the points $\boldsymbol{z}$ form a continuum we have
the following probability density 
\begin{equation}
P_{Z_{n}}(\boldsymbol{z})=\overline{\delta\left(z_{1}-\mathrm{Re}(Z_{n}(t))\right)\delta\left(z_{2}-\mathrm{Im}(Z_{n}(t))\right)}\,.\label{eq:density}
\end{equation}
In particular, the even moments of the distance of the walker from
the origin 
\begin{equation}
I_{m}:=\overline{\left|\chi(t)\right|^{2m}}\label{eq:moments}
\end{equation}
have been the subject of intense research efforts~\cite{borwein_arithmetic_2011,Borwein_Straub_Wan_Zudilin_Zagier_2012,garcia-garcia_exact_2018,legramandi_moments_2024}. As we will see, these quantities
correspond to the long time plateau studied in random matrix theory.
More precisely, in the context of RMT the following quantity has been
studied: $\mathsf{E}_{x}\left[\left|\chi_{x}(t)\right|^{2m}\right]$,
where we indicate with $x$ the collection of random variables, $\chi_{x}$
denotes that $\chi$ depends on the particular realization $x$, and
$\mathsf{E}_{x}\left[\bullet\right]$ is the ensemble average. The
connection between $I_{m}$ of Eq.~(\ref{eq:moments}) and $\mathsf{E}_{x}\left[\left|\chi_{x}(t)\right|^{2m}\right]$
studied in RMT is given by the following result (see Appendix \ref{sec:Proof-of-commutativity}
for a proof and Sec.~4.5 of \citep{oliviero_random_2021} for a similar
result):
\begin{thm}
\label{thm:commutativity}Let $H(x)$ be a Hamiltonian dependent on
the random variables $x$. If the induced distribution $\mu_{\boldsymbol{E}}(d\boldsymbol{E})$
of the eigenvalues of $H(x)$ is absolutely continuous,  and $\tr\left(\Pi_{j}(x)\right)$
and $d_{j}(x)$ do not depend on $x$, the ensemble average and time
average commute and one has:
\begin{align}
\lim_{t\to\infty}\mathsf{E}_{x}\left[\left|\chi_{x}(t)\right|^{2m}\right] & =\overline{\mathsf{E}_{x}\left[\left|\chi_{x}(t)\right|^{2m}\right]}\nonumber \\
= \mathsf{E}_{x}\left[\overline{\left|\chi_{x}(t)\right|^{2m}}\right] & =\overline{\left|\chi_{x}(t)\right|^{2m}}.
\end{align}
\end{thm}
The last equation means that infinite time moments are independent
of the specific realization $x$. All the hypotheses are very natural
for random models, in particular if degeneracies exist, they should
not depend on the random variables. This result applies, for example,
to the SYK model at infinite temperature ($\rho=\1$) \footnote{In principle we do not know if the eigenvalue distribution is absolutely
continuous. Several numerical calculations and the computations in \citep{garcia-garcia_exact_2018}
suggest that this is the case.} and means that the value of the moments of the SFF, the plateaus,
can be obtained also taking the infinite time average and even over
a single realization. 

So far, the random-walk interpretation of the spectral form factor
provides a perhaps amusing physical analogy, but
there are other tools that can be used to analyze a random walk. In
general, the trajectory formed by a random walk is a fractal: a geometrical
object whose Hausdorff dimension can assume non-integer values. To
gain some insight let us consider the familiar random walk on a two-dimensional
square lattice. At each step, the walker chooses one out of the four
possible directions uniformly at random. In the scaling limit, where
the number of steps  goes to infinity and the lattice spacing goes to zero (in a precise way\footnote{Changing variables to $s=n\tau$, $\boldsymbol{x}(s)=aZ_{n}$,
($s$ has units of time and $\boldsymbol{x}$ of space) in the scaling
limit $a,\tau\to0$ with $a^{2}/\tau=2D$ constant, $P_{Z_{n}}\left(\boldsymbol{z}\right)$
becomes $P(\boldsymbol{x},s)=\left(4\pi Ds\right)^{-1}\exp(-\left\Vert \boldsymbol{x}\right\Vert ^{2}/(4Ds))$,
the familiar Green's function of diffusion (or heat) equation. See e.g.~\cite{parisi_statistical_1998}}),
the random walk becomes a Wiener process
\footnote{Random walk, Brownian motion and Wiener process are sometimes all
considered synonym. In 1827 Brownian motion was observed by botanist
R.~Brown in minute particles suspended in liquids. Wiener constructed
the ``Wiener process'' in 1923 as a mathematical model of Brownian
motion.}, for which it can be proven that the fractal dimension
is two with probability one \citep{falconer_geometry_1985}. Another
interesting geometrical feature of a fractal is its frontier, which
is defined as follows. Let $K$ be a compact, connected set of the
plane. The complement of $K$ has one unbounded component and its
boundary is the frontier of $K$. Intuitively the frontier of a set
is the boundary without the inner islands. Mandelbrot conjectured
in 1982 that the fractal dimension of the frontier of the Wiener process
is 4/3 \citep{mandelbrot_fractal_1982}. This was rigorously proved
in 2001 using methods of stochastic Loewner evolution \citep{lawler_dimension_2001}.
Our aim here is to study the pair $\left(H,\rho\right)$ using methods
of fractal geometry, in particular studying the dimension of the frontier
of its corresponding random walk. As we will show, for independent
spectra $\left\{ E_{j}\right\} $ and if the weights $d_{j}$ satisfy a certain \emph{Lyapunov} condition
to be defined momentarily, the scaling limit of its associated random
walk becomes a Wiener process in the thermodynamic limit. We begin
by computing the probability distribution of the position Eq.~(\ref{eq:density}),
or rather its renormalized form $Y_{n}=Z_{n}/\Delta Z_{n}$ where
$\left(\Delta Z_{n}\right)^{2}=\overline{\left|Z_{n}(t)\right|^{2}}$,
which can be readily shown to be equal to $\left(\Delta Z_{n}\right)^{2}=\sum_{j=1}^{n}d_{j}^{2}$. 

We say that the spectrum of $H$ is \emph{independent} if the energies
$\left\{ E_{j}\right\} $ are linearly independent over the rationals.
Proving that a certain set of numbers is independent can be a mathematically
daunting task. As an extreme example, it is not yet known whether the single-qubit spectrum, $\left\{ 1,\gamma_C \right\} $, where $\gamma_C=0.5772\ldots $ is the 
Euler-Mascheroni constant~\cite{Weisstein_EulerMascheroni}, is independent or not. However, in general,
adding a tiny amount of randomness makes the spectrum independent
with probability one \cite{deutsch_quantum_1991}. For example, the classical ensembles of random
matrices, such as GUE, GOE and so on, have independent spectrum almost surely. If the energies
are independent, as a consequence of the Kronecker-Weyl ergodic theorem (see \cite{arnold_mathematical_1989}),  as $t$ increases, the phases $tE_{j}$ in Eq.~(\ref{eq:chit})
fill the torus $\mathbb{T}^{n}$ densely and uniformly, so $Z_{n}$
becomes a sum of $n$ complex independent but not necessarily identically distributed random variables whose phases
are distributed uniformly in $[0,2\pi)$. 
Under a certain \emph{Lyapunov} condition that ensures that a sufficient number of weights
$d_{j}$ are of the same order of magnitude,  the central limit
theorem (CLT)~\cite{billingsley_probability_2012} can be established, and the random variable $Y_n$ becomes Gaussian in the limit $n\to \infty$. This is precisely the content of the standard Lyapunov CLT for independent but not identically distributed random variables, that fits particularly well to our situation (see e.g.~\cite{billingsley_probability_2012}). The proof for our specific case is simple enough that we prefer to provide it here in order to keep the exposition self-contained (see  Appendix \ref{sec:Proof-of-CLT}).  
Defining 
\begin{equation}
R_{q}^{n}:=\frac{\sum_{j=1}^{n}d_{j}^{2q}}{\left(\sum_{j=1}^{n}d_{j}^{2}\right)^{q}}\,,
\end{equation}
we say that the weights $\left\{ d_{j}\right\} $ satisfy the \emph{Lyapunov} condition, if there exist a $q>1$ such that
\begin{equation}
\lim_{n\to\infty}R_{q}^{n}=0\,. \label{eq:generic_weights}
\end{equation}
Note that, for example, for $\rho=\1$ and if the Hamiltonian spectrum
is non-degenerate $d_{j}=1\,\forall j$, one gets $R_{q}^{n}=n^{-(q-1)}$
and Eq.~(\ref{eq:generic_weights}) holds trivially for all $q>1$.
For generic spectra, the weights $d_{j}$ control the variance of
each independent variable. Each variance, however, depends (possibly)
also on $n$ through $L$. This introduces a sort of ``correlation''
among the random variables and one cannot simply use the standard CLT for independent, identically distributed random variables.
However, one can show that the Lyapunov condition Eq.~(\eqref{eq:generic_weights}) implies $R_{p}^{n}\to0$
for all $p>1$, which in turn implies that all the cumulants of $Y_{n}$
beyond the second tend to zero implying Gaussianity. 
The details are shown in Appendix \ref{sec:Proof-of-CLT}.  
\begin{thm}
\label{thm:CLT_lyapunovplus} For independent spectrum $\left\{ E_{j}\right\} $
if the weights $\left\{ d_{j}\right\} $ satisfy the Lyapunov condition Eq.~\eqref{eq:generic_weights}, the multivariate
probability distribution of the rescaled variable $Y_{n}$ tends to
a standard 2D Gaussian in the thermodynamic limit, i.e.,
\[
\lim_{L\to\infty}P_{Y_{n}}(\boldsymbol{y})=\frac{1}{\pi}e^{-\left\Vert \boldsymbol{y}\right\Vert ^{2}}.
\]
The converse also holds: if the weights do not satisfy the Lyapunov condition Eq.~(\ref{eq:generic_weights}) the limiting
distribution is not Gaussian. 
\end{thm}
Theorem~\ref{thm:CLT_lyapunovplus} has several consequences: i) The moments of the squared
renormalized distance $\overline{\left|Y_{n}\right|^{2q}}$ approach
$q!$ in the thermodynamic limit, and, which is the same: ii) The
distribution of the square distance tends to $P_{\left|Y_{n}\right|^{2}}(r)=\vartheta(r)e^{-r}$,
i.e.,~$\left|Y_{n}\right|^{2}\stackrel{L\to\infty}{\longrightarrow}\mathsf{Exp}(1)$.
iii) The distance itself follows a Rayleigh distribution $P_{\left|Y_{n}\right|}(\rho)\to\vartheta(\rho)2\rho e^{-\rho^{2}}$
(with moments $\overline{\left|Y_{n}\right|^{q}}=\Gamma(1+q/2)$).
iv) At finite size the unscaled variable $Z_{n}$ is approximately
distributed as 
\[
P_{Z_{n}}\left(\boldsymbol{z}\right)\simeq\frac{1}{\pi\Delta Z_{n}^{2}}e^{-\left\Vert \boldsymbol{z}\right\Vert ^{2}/(\Delta Z_{n}^{2})}.
\]
v) If all the weights $d_{j}=1$, $\Delta Z_{n}^{2}=$number of steps$=n$.
This is the essential ingredient of the Brownian motion/Wiener process:
the average distance at time $n$ is proportional to $\sqrt{n}$. 

There are essentially two ways to violate the Lyapunov condition:
a) one can have $\left|\left\{ d_{j} > \epsilon \right\} \right|=M$ where $\epsilon$
is a (tiny) threshold and $M$ a ``small'' number. This implies that $Y_{n}$, in practice, is a sum of only $M$
independent variables and the CLT is violated; b) One can have the
variances of each random variable conspire such that the cumulants
of $Y_{n}$ tend to a non-zero value. Condition a) happens, for example, when $\rho=e^{-\beta H}/Z$ at very low temperature when only the lowest energy states are populated. We show evidence of this effect in the SYK model in Appendix~\ref{sub:SFF_SYK}. At high temperature the distribution of  $\vert Y_n \vert^2$ is $\mathsf{Exp}(1)$ in accordance with Theorem~\ref{thm:CLT_lyapunovplus} while at very low temperature the distribution becomes double peaked signaling a violation of the Lyapunov condition~\eqref{eq:generic_weights}.   
Condition b) happens, for example,
for $d_{j}=\left(j/n\right)^{-\alpha}$for $\alpha>1/2.$ In this
case the ratios converge to universal constants $R_{q}^{n}\to\zeta(2q\alpha)/\zeta(2\alpha)^{q}$
where $\zeta(z)$ is the Riemann zeta function (see \citep{campos_venuti_universal_2014}
for more details). This situation is (highly) unlikely when
$d_{j}$ is simply the degeneracy of level $E_{j}$ ($\rho=\1$),
but becomes possible at finite temperature $\rho=e^{-\beta H}$. In
case $\rho=|\psi\rangle\langle\psi|$, such violation of the CLT is obtained by a  {\it small quench} performed at
a quantum critical point of $H$ 
as was detailed in \citep{campos_venuti_universal_2014,campos_venuti_theory_2015,campos_venuti_theory_2015-1}. Very briefly, the small quench condition is the following. Let $H(\lambda)= H_0 + \lambda V$ be a Hamiltonian describing a system of finite linear size $L$ with a quantum critical point at $\lambda_c$. The ground state of $H(\lambda_c)$ is precisely $|\psi\rangle $ and the Hamiltonian under consideration is $H(\lambda)$ with $\delta \lambda = |\lambda - \lambda_c | \ll L^{-1/\nu }$ where $\nu $ is the correlation length critical exponent. More details can be found in \citep{campos_venuti_universal_2014,campos_venuti_theory_2015,campos_venuti_theory_2015-1}. 

Note also that when $d_j$ is the degeneracy of level $j$, both cases above are extremely unlikely in RMT. For example, $d_j = 1$ for the classical random matrix ensembles (GUE, GOE, etc.) with probability one, while for the SYK model $d_j =1,(2) $ depending on whether $ N \mod 8 =0 $ (or not) (see Appendix \ref{sec:numerics}).

Motivated by the remark v) above, we now investigate if and under what
conditions, the above random walk becomes a Wiener process in the
scaling limit. The fact that a sum of independent, identically distributed random variables, when properly rescaled, converges to a Wiener process is a standard result that goes under the name of Donsker's theorem or functional CLT. In our case, the variables are independent but not identically distributed. The extension of Donsker's theorem to the not identically distributed setting was first pursued in \cite{borovkov_estimates_1972,borovkov_estimates_1981,borovkov_rate_1983}. Using the known characterization of the Wiener process \cite{falconer_geometry_1985}, we provide a particularly simple proof for our case. 

We define the scaling limit by considering
the following random process:
\begin{equation}
W_{s}^{N}=\frac{1}{\Delta Z_{N}}\sum_{j=1}^{\left\lfloor Ns\right\rfloor }d_{j}e^{-itE_{j}}\,,\label{eq:wiener_N}
\end{equation}
where $s$ is the time variable of the random process. Note that if
we fix $N=\alpha N_{B}$ then $s$ is constrained to $[0,1/\alpha]$
but other choices are possible that result in $s$ being defined on
a larger set. To check whether $W_{s}^{N}$ defines a two dimensional
Wiener process $W_{s}$ in the limit $N\to\infty$, one should check
(see e.g.~\citep{falconer_geometry_1985}): a) That the paths have
\emph{independent increments}, that is, $W_{s_{2}}-W_{s_{1}}$ and
$W_{s_{4}}-W_{s_{3}}$ are independent if $s_{1}\le s_{2}\le s_{3}\le s_{4}$;
and b) That the distribution of $W_{s+h}-W_{s}$ is \emph{stationary}
(does not depend on $s$), \emph{isotropic }(does not depend on direction)
and is Gaussian with zero mean such that
\begin{equation}
\mathrm{Prob}\left(W:\left|W_{s+h}-W_{s}\right|\le\rho\right)=\frac{1}{h}\int_{0}^{\rho}dr\,r\exp\left(-r^{2}/(2h)\right).\label{eq:conditionb}
\end{equation}
We now define the "time dependent" version of $R^n_q$: 
\begin{equation}
    R_{1}^{N}(h,s):=\frac{\sum_{k=\left\lfloor Ns\right\rfloor +1}^{\left\lfloor N(s+h)\right\rfloor }d_{k}^{2}}{\sum_{k=1}^{N}d_{k}^{2}} \, .  
\end{equation}
We have the following result.
\begin{thm} \label{thm:Wiener}
If the spectrum $\left\{ E_{j}\right\} $ is independent and the weights
$\left\{ d_{j}\right\} $ satisfy the Lyapunov condition Eq.~\eqref{eq:generic_weights} and additionally satisfy 
\begin{equation}
\lim_{N\to\infty}R_{1}^{N}(h,s)=h \, ,\label{eq:Lyapunov_wiener}
\end{equation}
then $W_{s}^{N}$ converges to a Wiener process as $N\to\infty$. 
\end{thm}
Indeed, condition a) above is satisfied if the spectrum is generic
while condition b) is satisfied provided the weights satisfy Eqs.~(\ref{eq:generic_weights}) and (\ref{eq:Lyapunov_wiener}) (see Appendix \ref{sec:proof_Wiener_process}). 

Theorems \ref{thm:CLT_lyapunovplus} and \ref{thm:Wiener} are generalizations to the many-body setting of similar results and conjectures formulated for one-body problems with or without chaotic classical counterparts \cite{marklof_spectral_1998,aurich_temporal_1999,leboeuf_riemannium_2001}. For example, in Ref.~\cite{aurich_temporal_1999} it has been put forward the conjecture that for bound quantum systems with chaotic  classical dynamics, the normalized version of $|Z_n|$ has a Rayleigh distribution in the infinite size limit, while for integrable systems the distribution shows a non-universal behavior. Conversely, it has been shown in \cite{marklof_spectral_1998} that, for particular integrable systems, the rectangular billiards, the distribution of the renormalized random walk $Z_n$ tends to a non-Gaussian limit distribution. 

The main contribution of theorems \ref{thm:CLT_lyapunovplus} and \ref{thm:Wiener}, is the introduction of the non-constant weights $d_j$ and consequently of the Lyapunov conditions Eqs.~(\ref{eq:generic_weights}) and \eqref{eq:Lyapunov_wiener}. As we will show in Section \ref{sec:numerics}, at the hand of the XXZ spin chain with next-nearest neighbor interactions, even for non-integrable chaotic systems, these conditions cannot be neglected. In particular, as we have already commented, they are violated at low temperature and for particular initial states $\rho$. 

\section{Exact Moments of the Spectral Form Factor}
After having assessed the Gaussian regime, we now give the exact expression
of SFF's moments, which include deviation from Gaussianity. To compute
the even moment $I_{M}$ a weaker hypothesis on the spectrum $\left\{ E_{j}\right\} $
than independence is sufficient. We say that the numbers $\left\{ E_{k}\right\} _{k=1}^{D}$,
assumed to be all different, satisfy the non-degeneracy condition
at order $M\in\mathbb{N}$ ($M$-ND) if, for any pair of sequences
$\alpha_{i},\beta_{i}\in\left\{ 1,2,\ldots,D\right\} $, $i=1,2,\ldots,M$,
the equation
\begin{equation}
\sum_{j=1}^{M}E_{\alpha_{j}}=\sum_{j=1}^{M}E_{\beta_{j}}\label{eq:M-ND-1}
\end{equation}
 implies that $\{\alpha_{1},\alpha_{2},\ldots,\alpha_{M}\}$ is a
permutation of $\{\beta_{1},\beta_{2},\ldots,\beta_{M}\}$, i.e.,~there
exist a permutation $\pi\in\mathbb{S}_{M}$ such that $\alpha_{i}=\beta_{\pi(i)}$
for $i=1,\ldots,M$. In Appendix \ref{sec:Non-degeneracy-and-independence}
we show that the condition $M$-ND for all $M\in\mathbb{N}$ is equivalent
to that of independent spectrum. 

As a remark, as we will see in detail below, for non-degenerate spectra notice that one can approximate $I_{m}\simeq D^{m}m!$. However, often in literature this is seemingly considered an exact result instead of an approximation \cite{legramandi_moments_2024,winer_hydrodynamic_2022}. This is manifestly wrong already
for $m=2$ even in the case where $d_{j}=1$, as an explicit calculation
assuming $2$-ND gives $I_{2}=2D^{2}-D$. However, if the weights satisfy the Lyapunov condition Eq.~\eqref{eq:generic_weights}, Theorem \ref{thm:CLT_lyapunovplus} shows that, indeed, at
leading order in $D$, $I_{m}\simeq D^{m}m!$. Again, this approximation fails if 
the weights do not satisfy the Lyapunov condition~\eqref{eq:generic_weights}, for example, in the case $d_{j}=\left(j/n\right)^{-\alpha}$
for $\alpha>1/2$ or $d_{j}=\tr\left(e^{-\beta H}\Pi_{j}\right)$
at very low temperature $1/\beta$. 

A correct formula for the case $d_{j}=1$ (and independent spectrum) was
found in the mathematical literature \citep{borwein_arithmetic_2011}.
Here we give an exact formula valid for general $d_{j}$ (see Appendix
\ref{sec:SFF_moments}):
\begin{thm}\label{thm:moments}
Assuming the energies satisfy non-degeneracy at order $p$, then,
for $k=1,\ldots,p$
\[
I_{k}=\overline{\left|\chi(t)\right|^{2k}}=\left(k!\right)^{2}\sum_{\sum_{i=1}^{N_{B}}k_{i}=k}\prod_{i=1}^{N_{B}}\left(\frac{d_{i}^{k_{i}}}{k_{i}!}\right)^{2}. \label{eq:moments0}
\]
Moreover, defining the coefficients $a_{n}$ via the series $\ln(I_{0}(2\sqrt{z}))=\sum_{n=1}^{\infty}a_{n}z^{n}/n!$
and $X_{n}:=\sum_{j=1}^{N_{B}}(d_{j})^{2n}$ one has the following
recursion
\begin{equation}
I_{p}=\sum_{q=1}^{p-1}\left(\begin{array}{c}
p-1\\
q-1
\end{array}\right)\frac{p!}{(p-q)!}a_{q}X_{q}I_{p-q}+p!a_{p}X_{p}.\label{eq:recursion}
\end{equation}
\end{thm}
This theorem offers an exact, iterative, closed-form expression for arbitrary SFF moments beyond the Gaussian approximation. For example, assuming $3$-ND, using Eq.~(\ref{eq:recursion}), one
easily obtains
\begin{align}
I_{1} & =X_{1},I_{2}=2(X_{1})^{2}-X_{2}\label{eq:moments1}\\
I_{3} & =6\left(X_{1}\right)^{3}-9X_{1}X_{2}+4X_{3}\,.\label{eq:moments2}
\end{align}
One recognizes a term $m!(X_{1})^{m}$ appearing in $I_{m}$ corresponding
to the Gaussian approximation. If the weights satisfy the Lyapunov condition Eq.~\eqref{eq:generic_weights}, this is the leading
term in the sense that $\lim_{L\to\infty}I_{m}/\left(X_{1}\right)^{m}=m!$
(and so $I_{m}=m!\left(X_{1}\right)^{m}+o\left(\left(X_{1}\right)^{m}\right)$),
but this is no longer true if the weights do not satisfy  Eq.~\eqref{eq:generic_weights}. 
In Appendix~\ref{sub:SFF_SYK} we show that formula Eq.~(\ref{eq:recursion}) always gives the correct moments for the SYK model, while the Gaussian approximation fails at low temperature. 

Lastly, we would like to mention that the sequence in Eq.~(\ref{eq:moments0}) arises in the study of families of Calabi-Yau varieties and the periods of these varieties can be obtained computing recursions for the $I_m$'s, see \cite{verrill_sums_2004,verrill_root_1996}. 

\section{Integrable quadratic models}
What happens when the Hamiltonian is integrable? For integrable quadratic
Hamiltonians, that is, when $H$ is quadratic in either Bosonic or
Fermionic creation/annihilation operators, the hypothesis of independent
spectrum cannot hold. In fact, the energies of these models have the
form $E_{\boldsymbol{n}}=\sum_{k=1}^{L}n_{k}\epsilon_{k}$ with $n_{k}=0,1$
($n_{k}=0,1,2,\ldots$) for Fermions (Bosons) so there can be at most
$L$ independent energies (e.g.~the one-particle energies). For concreteness,
we will stick to the Fermionic case in the following. The value of
$\chi(t)$ is completely specified by the vector $\varphi(t):=(E_{1}t,E_{2}t,\ldots,E_{N_{B}}t)\mod{2\pi}$
($\chi(t)=\sum_{j=1}^{N_{B}}d_{j}e^{-i\varphi_{j}(t)}$). For models
with free Fermionic spectra, the closure of the trajectory $\varphi(t)$
as $t$ increases is a torus $\mathbb{T}^{M}$ where $M\le L$ is
the number of independent energies. A relation among the energies
is an equation of the form $\sum_{j}a_{j}E_{j}=0$ with integer $a_{j}$.
So, quadratic models have at least $D-L$ relations among the energies.
Such a large number of relations implies that the steps of the walker
are not even approximately independent. In this case we could not
compute the probability distribution of the walker Eq.~(\ref{eq:density})
but with mild assumptions we obtained the distribution of $\left|\chi(t)\right|^{2}$ (note that $\chi(t)=Z_{N_{B}}(t)$).
We have the following result:
\begin{thm}
\label{thm:free_fermion}Let $H$ be defined on a $2^{L}$ dimensional
Hilbert space, have spectrum $E_{\boldsymbol{n}}=\sum_{k=1}^{L}n_{k}\Lambda_{k}+E_{0}$
with $n_{k}=0,1$. We allow the one-particle spectrum $\left\{ \Lambda_{k}\right\} _{k=1}^{L}$
to have degeneracies, let $\left\{ \epsilon_{j}\right\} _{j=1}^{L_{B}}$
be the different energies with degeneracy $g_{j}$. For simplicity
we consider the infinite temperature case, $\beta=0$, i.e.,~$\chi(t)=\tr\left(e^{-itH}\right)$.
If the set $\left\{ \epsilon_{j}\right\} _{j=1}^{L_{B}}$ is independent
and the degeneracies are generic, i.e.,~there exist a $q>2$ such
that
\begin{equation}
\lim_{L\to\infty}S_{q}^{L_{B}}=0,\,\,\,\mathrm{with}\,\,\,S_{q}^{L_{B}}:=\frac{\sum_{j=1}^{L_{B}}g_{j}^{q}}{\left(\sum_{j=1}^{L_{B}}g_{j}^{2}\right)^{q/2}},\label{eq:generic_integrable}
\end{equation}
then for $L\to\infty$ the random variable $Y:=\ln\left|\chi(t)\right|^{2}/\sqrt{\sum_{j=1}^{L_{B}}g_{j}^{2}}$
converges in distribution to a Gaussian, more precisely
\[
\frac{\log\left(\left|\chi(t)\right|^{2}\right)}{\sqrt{\sum_{j=1}^{L_{B}}g_{j}^{2}}}\to\mathsf{N}\left(0,\frac{\pi^{2}}{3}\right).
\]
Alternatively $\left|\chi(t)\right|^{2/\sqrt{\sum_{j=1}^{L_{B}}g_{j}^{2}}}\to\mathsf{LogNormal}\left(0,\pi^{2}/3\right)$.
Note that $\pi^{2}/3$ is the variance of $\mathsf{Unif}\left[0,2\pi\right]$. 
\end{thm}
Remark. For $\beta\neq0$ the result continues to hold provided Eq.~(\ref{eq:generic_integrable})
is satisfied. Namely $Y$ still converges in distribution to $\mathsf{N}(0,\sigma)$
but  we could not evaluate the variance $\sigma$ (which now depends
on $\beta$). However, as for the non-integrable case, we can anticipate
that for sufficiently low temperature (large $\beta$) Eq.~(\ref{eq:generic_integrable})
will break down. The proof is yet another form of central limit theorem
but in this case $Y$ is a product of independent random variables,
see Appendix \ref{sec:distribution_free}. 

For completeness, we also give the expression for the moments in this
case. 
\begin{thm}
With the same setting and hypotheses as in theorem \ref{thm:free_fermion},
one has
\[
\overline{\left|\chi(t)\right|^{2M}}=\prod_{j=1}^{L_{B}}\left(\begin{array}{c}
2g_{j}M\\
g_{j}M
\end{array}\right). \label{eq:moments_free}
\]
\end{thm}

Consider for simplicity the non-degenerate case ($g_j=1, L_B=L$). From Eq.~(\ref{eq:moments_free}) we obtain $\overline{|\chi(t)|^2}=2^L=D$ which is the same result as for generic, non-integrable walkers. One may then be led to conjecture that quasi-free fractals are of the same size as non-integrable ones. However,  analytically continuing Eq.~(\ref{eq:moments_free}) to $M=1/2$ we obtain that the average distance of the quasi-free walker is $\overline{|\chi(t)|}=(4/\pi)^L=D^{0.348}$. This number is exponentially smaller in $L$ than in the non-integrable case. In fact, from point iii) after Theorem~\ref{thm:CLT_lyapunovplus} we get $\overline{|\chi(t)|} \simeq D^{1/2}\Gamma(3/2) = \sqrt{2}^L \sqrt{\pi}/2$. This means that the average size of quasi-free fractals is smaller than that of non-integrable ones. This feature is also apparent in Fig.~\ref{fig:mainplot} where, on the same scale, the fractal on the right (quasi-free) appears smaller in extension than the one on the left (non-integrable). 

\section{Numerical estimates of the dimension of the fractal frontier}
\label{sec:numerics}
We now come to the comparison of the previous theorems to realistic
physical models for the case $\rho=e^{-\beta H}$. It is natural to
expect that non-integrable systems have independent spectrum and that
the weights satisfy the Lyapunov condition Eq.~\eqref{eq:generic_weights} at infinite or sufficiently high temperature
$1/\beta$. Instead, when the temperature is very low, the ground
state $d_{1}=e^{-\beta E_{1}}\tr\left(\Pi_{1}\right)$ dominates and
one can violate the CLT. On the other hand, quasi-free Fermionic systems
have only $L_{B}\le L=O\left(\ln D\right)$ independent energies (the
one-particle energies). So, for high temperature, looking at the distribution
$\left|\chi(t)\right|^{2}$ we expect $\mathsf{Exp}(1)$ in the non-integrable
case and \textsf{$\mathsf{LogNormal}$} for quasi-free Fermionic spectra.
What about Bethe-Ansatz (BA) solvable models? One can make an argument
for both cases. Indeed, say we consider the XXZ chain for concreteness.
BA integrability implies that the spectrum (in each conserved sector)
depends on $O\left(\ln D\right)$ quantities, the rapidities. However
it is not clear if these quantities are combined to form the many
body-spectrum in a way that does not introduces relations among the
energies. 
We perform our numerics on the following XXZ chain next-to-nearest-neighbor coupling $\alpha$
\begin{align}
H & =\sum_{j=1}^{L}\left(\sigma_{j}^{x}\sigma_{j+1}^{x}+\sigma_{j}^{y}\sigma_{j+1}^{y}+\Delta\sigma_{j}^{z}\sigma_{j+1}^{z}\right) \nonumber \\
 & +\alpha\sum_{j=1}^{L}\left(\sigma_{j}^{x}\sigma_{j+2}^{x}+\sigma_{j}^{y}\sigma_{j+2}^{y}+\Delta\sigma_{j}^{z}\sigma_{j+2}^{z}\right). \label{XXZNNNHamiltonian}
\end{align}
The model admits quasi free (for $\alpha = \Delta = 0$), Bethe-Ansatz ($\alpha  = 0$), and  non-integrable ($\alpha \neq 0$)
phases. The model has several conserved quantities, and the pattern of degeneracies $d_j = \tr (\Pi_j )$ is highly non-trivial, leading to $d_j$ as high as 8 for a sensible fraction of $j$ for $L=10$. Our results confirm that, 
at infinite temperature, the distribution of the normalized $\left|\chi(t)\right|^{2}$
is $\mathsf{Exp}(1)$ in the non-integrable phase and \textsf{$\mathsf{LogNormal}$}
in the quasi-free one. We also observe an $\mathsf{Exp}(1)$ distribution
for the BA case (see App.~\ref{appedix:numericalcheck}).

Regarding the fractal frontier dimension, the box-counting procedure~\cite{falconer_geometry_1985, Sokolovic_Mali_Odavic_Radosevic_Medvedeva_Botha_Shukrinov_Tekic_2017} gives
\(d_F = 1.32 \pm 0.08\) for the non-integrable XXZ model, \(d_F = 1.01 \pm 0.04\)
for the XX free (quasi-free) spin chain, and \(d_F = 1.24 \pm 0.08\) for the
Bethe–Ansatz (BA) solvable XXZ model (see App.~\ref{sec:appfractral} for
methodological details).  The quoted uncertainties are standard errors of the
mean computed over the sample ensemble. To assess parameter dependence we
estimated \(d_F\) on ensembles of 20 independent random-walk realizations for
each phase (quasi-free, BA, and non-integrable (NI)). For BA and NI the
results are displayed at two representative values of the control parameters
(\(\Delta\) for BA, \(\alpha\) for NI). The quasi-free chain is clearly
distinct (\(d_F\approx 1\)), while the NI estimates show negligible dependence
on \(\alpha\) and are consistent, within statistical uncertainty, with the
Brownian expectation \(d_F = 4/3\). The BA estimates are largely
\(\Delta\)-independent but lie below \(4/3\); with the current
ensemble size (20 realizations per point) we cannot conclusively determine
whether the BA case approaches \(4/3\) in the asymptotic limit or rather
converges to a different universal value slightly smaller than \(4/3\). The data obtained are shown in Fig.~\ref{fig:numerics_d_F}.

At this point, we cannot claim
for sure if, in
the BA case,  the discrepancy from $4/3$ is due to the complexity
of estimating the fractal frontier or genuinely $d_{F}\neq4/3$. The latter scenario would mean that the number of relations
among the energies for BA integrable models is not sufficient to violate
the CLT for the random variable $\left|\chi(t)\right|^{2}$ but can be
detected by looking at the frontier of the corresponding random walk.
The results are summarized in Table \ref{tab:Expectation-of-different}.

\begin{figure}[]
\begin{centering}
\includegraphics[width=8.5cm]{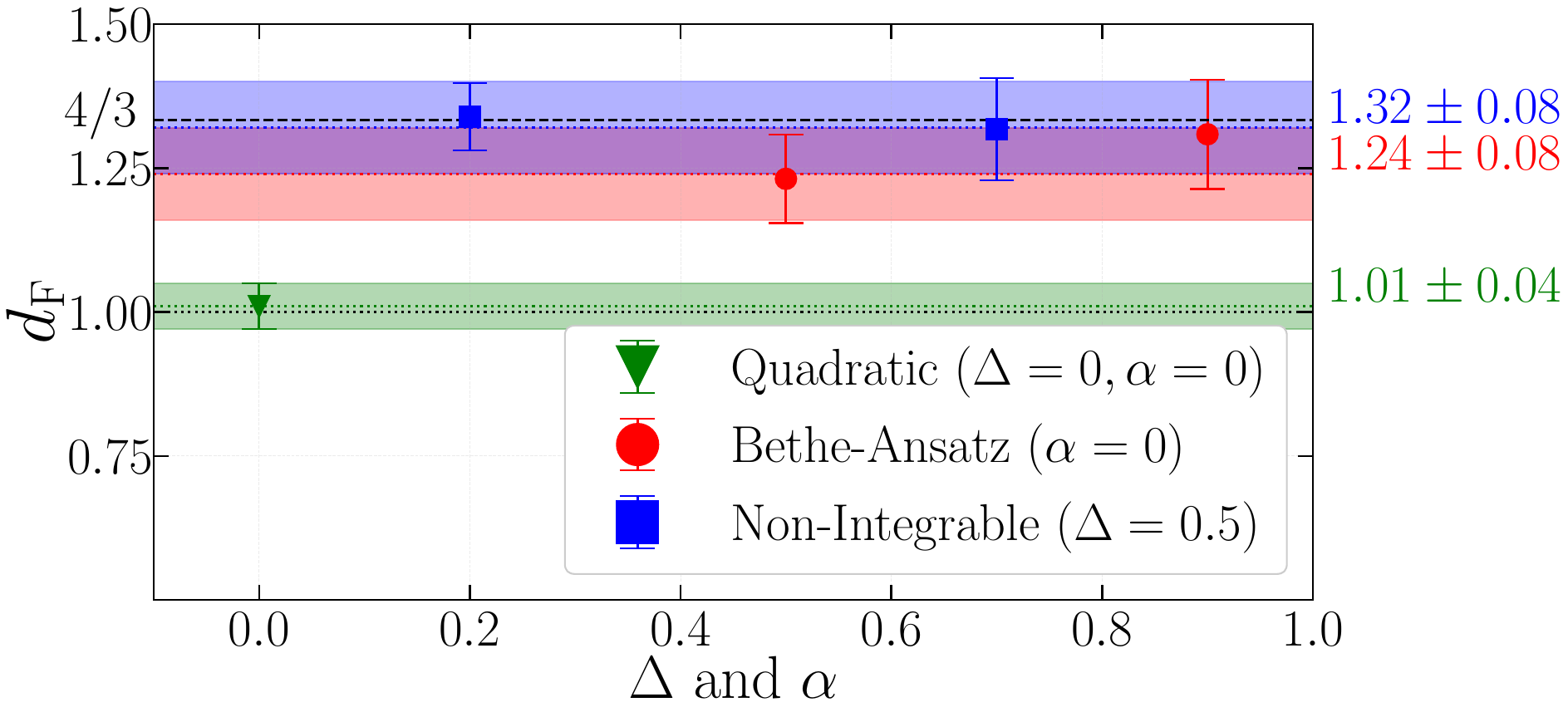}
\par\end{centering}
\caption{Numerical estimates of the Hausdorff dimension of the fractal frontier for quadratic, Bethe-Ansatz and non-integrable models. The green, red, and blue stripes are the average values within one standard deviation confidence of the numbers $d_F$ obtained for the quadratic, Bethe-Ansatz and non-integrable phases, respectively. In each phase we used $20$ samples to obtain the statistics. To allow for comparisons all data are obtained for size $L=12$. More details can be found in Appendix \ref{sec:appfractral}. } \label{fig:numerics_d_F}
\end{figure}

\section{Discussion and conclusions}

When we are interested in understanding a quantum model, it is customary to consider functions of the model Hamiltonian $H$, which, depending on context, carry particular information on the problem at hand. One of the most versatile of such functions is $\tr e^{-z H}$. 
When $z$ is real, $\tr e^{-z H}$  is the partition function and has information on the thermodynamics of the system. For complex $z=\beta+i t$, its modulus square is known as the spectral form factor (SFF), and is important in a variety of fields ranging from quantum chaos to black-hole dynamics \cite{prosen_time_1998,prosen_ergodic_1999,cotler_black_2017,CotlerErratum2018,caceres_spectral_2022}. In this context, usually an average is performed further on the SFF over an ensemble of Hamiltonians. Typical examples are the (random) SYK model or the classical matrix ensembles such as GUE or GOE.  Alternatively, in the spirit of statistical mechanics, where one observes the system for a sufficiently long time, one can average uniformly over time. It turns out that the two averages are related and using the latter, as already noted in \cite{prange_spectral_1997}, the SFF can be understood as a random walk. 

In this paper, we push this analogy further. We generalize the SFF to $\chi(t) = \tr \rho e^{-itH}$ for any positive operator $\rho$, and we show how to associate any pair $(\rho,H)$ with a random walk on the plane. For a Hamiltonian with spectral resolution $H=\sum_{j=1}^{N_B} E_j \Pi_j$ ($E_j$ eigenenergies, $\Pi_j$ spectral projectors, $N_B$ number of blocks), at step $n$ the walker rotates clockwise by an angle $tE_n$ and steps forward by an amount $d_j=\tr(\rho \Pi_n)$. When $t$ is a  random variable uniformly distributed on the half-line, this is a random walk, i.e., ~a fractal, and we study the pair $(\rho,H)$ through properties of this fractal geometry, in particular, we compute the Hausdorff dimension of its frontier. From this point of view, the moments of the SFF which correspond to the plateaus of \cite{cotler_black_2017,CotlerErratum2018} and have been studied extensively, correspond to moments of the distance of the walker from the origin at step $N_B$. 

We show that, under two assumptions, the random walk becomes a Wiener process in the thermodynamic limit, for which the dimension of its frontier has been famously proved to be $d_F = 4/3$ using Schramm-Loewner evolutions methods \cite{lawler_dimension_2001}. The assumptions are i) that the energy levels are linearly independent over the set of rationals, and ii) a technical, Lyapunov-like condition. For $\rho = \1$ and non-degenerate spectra, the Lyapunov condition is automatically satisfied. Via numerical simulations, we show that, for $\rho = e^{-\beta H}$ and non-integrable $H$, the Lyapunov condition is satisfied at sufficiently high temperature $1/\beta$ (but notably fails at low temperature and in other documented cases). 
In these cases, we observe numerically that the dimension of the fractal frontier is approximately $d_F \simeq 4/3$, thus confirming the hypothesis of independence of the energy spectrum. 
The analogous result for the SFF is that the distribution of $\vert\chi(t)\vert^2$ properly normalized becomes \textsf{Exp(1)}. This corresponds to $\overline{\vert\chi(t)\vert^{2m}} \simeq m!
\left ( \overline{\vert\chi(t)\vert^{2}} \right)^m$ a result widely known in the literature and sometimes called Gaussian approximation. Even if the hypothesis of independence of the energy levels can be proven, like in random matrix or SYK models, the Gaussian approximation hinges on the Lyapunov condition. This condition is violated at low temperature where only the lowest energy levels get effectively populated\footnote{To be precise, this feature is stable as the size approaches the thermodynamic limit for models with a gap above the ground state. For models with gapless excitations this feature appears in a crossover region for finite size}. On the other hand, we compute exactly the moments of the SFF without the Gaussian approximation solving the classical problem of a random walker taking steps of unequal lengths. The results agree with numerical experiments on the SYK model.

For free, integrable models, the condition of independent spectrum fails spectacularly as the number of independent energies is at most $N$ where $N$ is the number of free modes, in spite of a Hilbert space dimension $D$ exponential in $N$. This means that there are at least $D-N$ relations among the energies of the form $\sum_{j=1}^D a_j E_j =0 $ with $a_j \in \mathbb{Z}$. In this case, even when the Lyapunov condition is satisfied, the random walk is no longer a sum of independent variables. The exponentially many relations among the energies imply a correlation among a large fraction of the random variables that are no longer independent. As a result, the central limit theorem breaks down, and the random walk is not a Wiener process in the thermodynamic limit. 
Summarizing, when the Lyapunov conditions are satisfied, a value for the dimension of the fractal frontier close to $4/3$ signals that the random variables that form the random walk are sufficiently independent, while a value away from $4/3$ signals the breakdown of the central limit theorem and dependent variables. This result provides a novel method to inquire the difficult problem of independence of random variables. For instance, a long standing problem is that of the 
Bethe-Ansatz (BA) solvable case. The study of the BA fractal shows that is likely to be smaller than $4/3$ and therefore consistent with 
a breakdown of the CLT and a novel, universal value, but this analysis requires further investigation. 

There are several directions that open up for future studies. On the one hand, one can consider other kind of quantum mechanical objects, such as entanglement Hamiltonians (obtained as minus the logarithm of a partial trace of a quantum state),  quantum states themselves, or more general evolutions which include decoherence and dissipation. The study of the interplay between the SFF and the operator entaglement spectrum \cite{mcdonough2025bridgingclassicalquantuminformation} is indeed particularly promising.
Additionally, other properties of fractals can be used or even fractals embedded in a larger dimensional space. One intriguing possibility is to study a particular class of quantum many-body Hamiltonians whose eigenstates can be characterized as doped stabilizer states~\cite{PhysRevA.110.062427}. The doping refers to a very gentle way of perturbing otherwise integrable Hamiltonians. These Hamiltonians, although non-integrable, feature classical algorithms efficient in computing their time evolution, thermal states, and low-energy eigenstates. The study of their SFF and the associated fractal can thus shed some light in the direction of understanding how integrable features may be not  immediately lost and give insight towards a quantum KAM theorem, sewing together the chaotic features of the spectrum and those of the eigenvectors of a Hamiltonian \cite{e23081073}.

\section*{Acknowledgments}
This research was funded by the Research Fund for the Italian Electrical System under the Contract Agreement "Accordo di Programma 2022–2024" between ENEA and the Ministry of the Environment and Energy Safety (MASE)- Project 2.1 "Cybersecurity of energy systems". AH, JO acknowledge support from the PNRR MUR project PE0000023-NQSTI. AH acknowledges support from the PNRR MUR project CN 00000013-ICSC. JO acknowledges ISCRA for awarding this project access to the LEONARDO super-computer, owned by the EuroHPC Joint Undertaking, hosted by CINECA (Italy) under the project ID: PQC - HP10CQQ3SR.

\section*{Data Availability}
The numerical codes, datasets, and plotting scripts used in this work are available at Zenodo~\cite{zenodo}.

\appendix

%
\section{Proof of Theorem \ref{thm:commutativity}}

\label{sec:Proof-of-commutativity}We consider a random Hamiltonian
$H(x)\!=\!\sum_{j=1}^{N_{B}}E_{j}(x)\Pi_{j}(x)$ where $x$ is the collection
of (static) random variables. By assumption $d_{j}$ do not depend
on $x$. Starting from $\chi_{x}(t)=\sum_{j=1}^{N_{B}}d_{j}e^{-itE_{j}(x)}$
, for any realization $x$, we get
\[
\left|\chi_{x}(t)\right|^{2}=\sum_{i,j=1}^{N_{B}}F_{i,j}e^{-it(E_{i}(x)-E_{j}(x))}
\]
having defined, for convenience, $F_{i,j}=d_{i}d_{j}$. Hence 
\begin{align}
I_{m}:=\left|\chi_{x}(t)\right|^{2m} & =\sum_{i_{1}\cdots i_{m}}\sum_{j_{1}\cdots j_{m}}F_{i_{1},j_{1}}\cdots F_{i_{m},j_{m}} \nonumber \\
 & \times\exp\left(-it\sum_{l=1}^{m}\left(E_{i_{l}}(x)-E_{j_{l}}(x)\right)\right)
\end{align}
which we write compactly as
\begin{equation}
\left|\chi_{x}(t)\right|^{2m}=\sum_{s,r}f(s,r)e^{-it(\boldsymbol{E}_{s}(x)-\boldsymbol{E}_{r}(x))}\label{eq:commut_start}
\end{equation}
having defined the multi indices $s=\left(s_{1},\ldots,s_{m}\right)$
corresponding to $\boldsymbol{E}_{s}(x)=E_{s_{1}}(x)+E_{s_{2}}(x)+\cdots E_{s_{m}}(x)$
and similarly for $r$ and $\boldsymbol{E}_{r}(x)$ together with
the coefficients $f(s,r)$. Note that, if the degeneracies $\tr \Pi_j(x)$ do not depend on $x$, the sequences $\left\{ s_{j}\right\} $
depend only on $N_{B}$ and not on $x$. It is useful to go to the
occupation number representation. For any sequence $\left\{ s_{i}\right\} _{i=1}^{m}$
we define the (Bosonic) occupation numbers:
\begin{equation}
n_{k}^{s}:=\sum_{i=1}^{m}\delta_{E_{k}(x),E_{s_{i}}(x)},\label{eq:occupation_mapping-1}
\end{equation}
so $n_{k}^{s}$ counts how many terms in $\left\{ E_{s_{i}}(x)\right\} _{i=1}^{M}$
have energy $E_{k}(x)$. We also define $\boldsymbol{n}^{s}=\left(n_{1}^{s},\ldots,n_{N_{B}}^{s}\right)$
and write 
\[
\boldsymbol{E}_{s}(x)=\sum_{k=1}^{N_{B}}n_{k}^{s}E_{k}(x)=\boldsymbol{n}^{s}\cdot\boldsymbol{E}(x).
\]
Taking the infinite time average we get $\overline{e^{-it(\boldsymbol{E}_{s}(x)-\boldsymbol{E}_{r}(x))}}=\delta_{\left(\boldsymbol{E}(x)\cdot\left(\boldsymbol{n}^{s}-\boldsymbol{n}^{r}\right)\right)}$
where $\delta_{X}$ is a Kronecker delta. If the distribution of the
eigenvalues $\boldsymbol{E}$ is absolutely continuous the set $\left\{ E(x)\right\} _{k=1}^{N_{B}}$
is rationally independent almost everywhere. So, for almost any $x$,
$\delta_{\left(\boldsymbol{E}(x)\cdot\left(\boldsymbol{n}^{s}-\boldsymbol{n}^{r}\right)\right)}=\delta_{\left(\boldsymbol{n}^{s}-\boldsymbol{n}^{r}\right)}$.
So we obtain
\[
\overline{\left|\chi_{x}(t)\right|^{2m}}=\sum_{s,r}f(s,r)\delta_{\left(\boldsymbol{n}^{s}-\boldsymbol{n}^{r}\right)}.
\]
Incidentally the above equation does not depend on $x$ anymore and
so we obtain also $\overline{\left|\chi_{x}(t)\right|^{2m}}=\mathsf{E}_{x}\left[\overline{\left|\chi_{x}(t)\right|^{2m}}\right]$.
We now take the average over $x$ of Eq.~(\ref{eq:commut_start}):
\begin{align}
\mathsf{E}_{x}\left[\left|\chi_{x}(t)\right|^{2m}\right] & =\sum_{s,r}f(s,r)\mathsf{E}_{x}\left[e^{-it(\boldsymbol{E}(x)\cdot(\boldsymbol{n}^{s}-\boldsymbol{n}^{r}))}\right]\nonumber \\
 & =\sum_{s,r}f(s,r)\delta_{\left(\boldsymbol{n}^{s}-\boldsymbol{n}^{r}\right)}\nonumber \\
 & +\sum_{s,r}f(s,r)\mathsf{E}_{x}\left[e^{-it(\boldsymbol{E}(x)\cdot(\boldsymbol{n}^{s}-\boldsymbol{n}^{r}))}\right]\delta_{\left(\boldsymbol{n}^{s}\neq\boldsymbol{n}^{r}\right)}.
\end{align}
Let us call $O(t)$ the last term above. If the distribution of $\boldsymbol{E}$
is absolutely continuous: $\mu\left(d\boldsymbol{E}\right)=P_{\boldsymbol{E}}(\boldsymbol{E})d\boldsymbol{E}$,
$\mathsf{E}_{x}\left[e^{-it(\boldsymbol{E}(x)\cdot(\boldsymbol{n}^{s}-\boldsymbol{n}^{r}))}\right]=\hat{P}_{\boldsymbol{E}}\left(t(\boldsymbol{n}^{s}-\boldsymbol{n}^{r})\right)$
is the Fourier transform of a function in $L^{1}$. By Riemann-Lebesgue
lemma $\lim_{t\to\infty}\hat{P}_{\boldsymbol{E}}\left(t(\boldsymbol{n}^{s}-\boldsymbol{n}^{r})\right)=0$
whenever $\boldsymbol{n}^{s}\neq\boldsymbol{n}^{r}$. So we get $\lim_{t\to\infty}O(t)=0$
which implies that $\overline{O(t)}=0$ which proves the claim. 

\section{Proof of Theorem \ref{thm:CLT_lyapunovplus}}

\label{sec:Proof-of-CLT}
Here we prove for our specific setting the Lyapunov central limit theorem, which concerns the sum of independent variables
with different distributions. The variance of each random variable, proportional to  
$d_{j}^{2}$, may depend also $n$ and  this prevents direct application of the CLT for identically distributed random variables. The dependence on $n$ introduces a sort of correlation
between the different random variables. As we shall see, it turns out that
in our case the Lyapunov condition is still both necessary and sufficient
for the Gaussianity of the normalized sum as $n\to\infty$. Let us
compute the characteristic function of the random variable $Y_{n}=Z_{n}/\Delta Z_{n}$
($\boldsymbol{k}=(k_{1},k_{2})$):
\[
\overline{e^{i\boldsymbol{k}Y_{n}}}=\overline{\exp\left(i\sum_{j=1}^{n}\frac{d_{j}}{\Delta Z}_{n}\left(k_{1}\cos(tE_{j})+k_{2}\sin(tE_{j})\right)\right)}.
\]
Under assumption of independence of the energies $E_{j}$ the time
average becomes the uniform average over the torus $\mathbb{T}^{n}$
and we obtain
\begin{align*}
\overline{e^{i\boldsymbol{k}Y_{n}}}= & \prod_{j=1}^{n}\int\frac{d\vartheta_{j}}{2\pi}\exp\left(i\frac{d_{j}}{\Delta Z}_{n}\left(k_{1}\cos(\vartheta_{j})+k_{2}\sin(\vartheta_{j})\right)\right)\\
= & \prod_{j=1}^{n}J_{0}\left(\frac{d_{j}}{\Delta Z}_{n}\left\Vert \boldsymbol{k}\right\Vert \right)
\end{align*}
where $J_{0}$ is the Bessel function of the first kind, analytic
in a neighborhood of zero. Let us define $\ln J_{0}(x)=\sum_{p=1}^{\infty}b_{p}x^{2p}$
(for $x$ small enough $J_{0}(x)>0$). Note that $b_{1}=-1/4$. Then
\[
\overline{e^{i\boldsymbol{k}Y_{n}}}=\exp\left(\sum_{j=1}^{n}\sum_{p=1}^{\infty}b_{p}\left(\frac{d_{j}}{\Delta Z}_{n}\right)^{2p}\left\Vert \boldsymbol{k}\right\Vert ^{2p}\right).
\]
Defining $R_{p}^{n}=\sum_{j=1}^{n}\left(d_{j}/\Delta Z_{n}\right)^{2p}$
the above is written as
\begin{equation}
\overline{e^{i\boldsymbol{k}Y_{n}}}=\exp\left(\sum_{p=1}^{\infty}b_{p}R_{p}^{n}\left\Vert \boldsymbol{k}\right\Vert ^{2p}\right).\label{eq:series}
\end{equation}

The Lyapunov condition states that there is a $q>1$ such that $R_{q}^{n}\to0$
as $n\to\infty$. We will show that the Lyapunov condition implies
that $R_{p}^{n}\to0$ for all $p>1$. First note that $\lim_{n\to\infty}R_{q}^{n}=0$
iff 
\begin{equation}
\lim_{n\to\infty}\frac{\left\Vert \boldsymbol{d}_{n}\right\Vert _{2q}}{\left\Vert \boldsymbol{d}_{n}\right\Vert _{2}}=0\label{eq:d_cumulant}
\end{equation}
where $\left\Vert \boldsymbol{x}\right\Vert _{p}=\left(\sum_{j=1}^{n}\left|x_{j}\right|^{p}\right)^{1/p}$
and $\boldsymbol{d}_{n}=(d_{1},d_{2},\ldots,d_{n})$. Now for $p\ge2q$
$\left\Vert \boldsymbol{x}\right\Vert _{p}\le\left\Vert \boldsymbol{x}\right\Vert _{2q}$
so Eq.~(\ref{eq:d_cumulant}) implies $R_{p}^{n}\to0$ for all $p>2q$.
Consider now $p\in(2,2q)$ and define $\boldsymbol{x}_{n}=\boldsymbol{d}_{n}/\left\Vert \boldsymbol{d}_{n}\right\Vert _{2}$
so that $\left\Vert \boldsymbol{x}_{n}\right\Vert _{2}=1$. For all
$r\in(0,1)$ define $\boldsymbol{x}^{r}:=(\left|x(1)\right|^{r},\ldots,\left|x(n)\right|^{r})$.
Then
\begin{align*}
\left\Vert \boldsymbol{x}^{r}\right\Vert _{2/r} & =\left(\left\Vert \boldsymbol{x}\right\Vert _{2}\right)^{r}\\
\left\Vert \boldsymbol{x}^{1-r}\right\Vert _{2q/(1-r)} & =\left(\left\Vert \boldsymbol{x}\right\Vert _{2q}\right)^{1-r}.
\end{align*}
By H\"older inequality \cite{bhatia_MatrixAnalysis_1977}, for all $p\in(2,2q)$,
\begin{equation}
\left\Vert \boldsymbol{x}\right\Vert _{p}\le\left\Vert \boldsymbol{x}^{r}\right\Vert _{2/r}\left\Vert \boldsymbol{x}^{1-r}\right\Vert _{2q/(1-r)}\label{eq:holder}
\end{equation}
as long as 
\[
\frac{1}{p}=\frac{r}{2}+\frac{1-r}{2q}
\]
which implies $r=(2q-p)/(p(q-1))$. This is indeed in $(0,1)$ for
$p\in(2,2q)$. Applying Eq.~(\ref{eq:holder}) to our sequence of
sequences $\boldsymbol{x}_{n}$ we get 
\[
\left\Vert \boldsymbol{x}_{n}\right\Vert _{p}\le\left(\left\Vert \boldsymbol{x}_{n}\right\Vert _{2q}\right)^{1-r}
\]
which implies $\left\Vert \boldsymbol{x}_{n}\right\Vert _{p}\to0$
also for $p\in(2,2q)$. Hence we obtain that $R_{p}^{n}\to0$ for
all $p>1$ as claimed, i.e.,~$\lim_{n\to\infty}R_{p}^{n}=\delta_{p,1}$.
To pass the limit inside the series in Eq.~(\ref{eq:series}) note
that $R_{q}^{n}\le1$ for all $q$ so $\left|b_{p}R_{p}^{n}\left\Vert \boldsymbol{k}\right\Vert ^{2p}\right|\le\left|b_{p}\right|\left\Vert \boldsymbol{k}\right\Vert ^{2p}=:M_{p}$
and $\sum_{p}M_{p}=\left|\ln J_{0}(\left\Vert \boldsymbol{k}\right\Vert )\right|<\infty$
and by Tannery's theorem \cite{whittaker_course_2021} we can exchange the limit with the series
and obtain 
\[
\lim_{n\to\infty}\overline{e^{i\boldsymbol{k}Y_{n}}}=e^{-\left\Vert \boldsymbol{k}\right\Vert ^{2}/4}.
\]
Fourier transforming we get for the probability distribution of $Y=\lim_{n\to\infty}Y_{n}$
\begin{align}
P_{Y}(\boldsymbol{y}) & =\int_{\mathbb{R}^{2}}\frac{d\boldsymbol{k}}{(2\pi)^{2}}e^{i\boldsymbol{k}\cdot\boldsymbol{x}}e^{-\left\Vert \boldsymbol{k}\right\Vert ^{2}/4}\\
 & =\frac{1}{\pi}e^{-\left\Vert \boldsymbol{y}\right\Vert ^{2}}.
\end{align}

The converse statement also trivially holds, i.e.,~if Eq.~(\ref{eq:d_cumulant})
is not satisfied for any $q>1$ the cumulants of $Y$ are all non-zero
and its distribution is not Gaussian. Given the joint distribution $P_{Y}(\boldsymbol{y})$ it is straightforward
to obtain the distribution of the distance square from the origin
$P_{\left|Y\right|^{2}}(r)$:
\begin{align*}
P_{\left|Y\right|^{2}}(r) & =\int_{\mathbb{R}^{2}}d\boldsymbol{y}\delta(r-\left\Vert \boldsymbol{y}\right\Vert ^{2})P_{Y}(\boldsymbol{y})\\
 & =2\pi\int_{0}^{\infty}\rho d\rho\delta(r-\rho^{2})\frac{1}{\pi}e^{-\left\Vert \boldsymbol{y}\right\Vert ^{2}}\\
 & =\vartheta(r)e^{-r}.
\end{align*}
Analogously one obtains $P_{\left|Y\right|}(u)=\vartheta(u)2ue^{-u^{2}}$.
Going back to the unnormalized variables one obtains approximately
$P_{\left|Z_{n}\right|^{2}}(v)\simeq\vartheta(\nu)\exp\left(-\nu/\Delta Z_{n}^{2}\right)/\Delta Z_{n}^{2}$. 
Note that the distribution of $\left|Z_{n}\right|^{2}$ can be written
exactly under the assumption of independent spectrum.
Following \citep{campos_venuti_universal_2014} (see Eq.~(F2))
\[
\mathrm{Prob}\left(\left|Z_{n}\right|^{2}<r^{2}\right)=\int_{0}^{\infty}d\rho\,rJ_{1}(r\rho)\prod_{j=1}^{n}J_{0}\left(d_{j}\rho\right).
\]
 Substituting $r^{2}=s$ and differentiating with respect to $s$
we obtain
\[
P_{\left|Z_{n}\right|^{2}}(r)=\frac{1}{2}\int_{0}^{\infty}d\rho\,\rho J_{0}(\sqrt{s}\rho)\prod_{j=1}^{n}J_{0}\left(d_{j}\rho\right),
\]
which is Eq.~(F3) of \citep{campos_venuti_universal_2014} after
simplification. 

\section{Convergence to the Wiener process}

\label{sec:proof_Wiener_process}Here we show that under certain conditions
the random walk defined by Eq.~(\ref{eq:wiener_N}) converges to
a Wiener process. The steps are similar to those of Appendix \ref{sec:Proof-of-CLT}.
As mentioned in the main text, if the energies are independent, the path
has independent increment. Let us compute the probability distribution
of the following random variable
\[
\Delta^{N}(s,h):=W_{s+h}^{N}-W_{s}^{N}=\frac{1}{\Delta Z_{N}}\sum_{j=\left\lfloor Ns\right\rfloor +1}^{\left\lfloor N(s+h)\right\rfloor }d_{j}e^{-itE_{j}},
\]

where $t$ is a random variable subject to infinite-time average.
Under assumption of independence of the energies its characteristic
function is
\[
\overline{e^{i\boldsymbol{k}\Delta^{N}(s,h)}}=\prod_{j=\left\lfloor Ns\right\rfloor +1}^{\left\lfloor N(s+h)\right\rfloor }J_{0}\left(\frac{d_{j}}{\Delta Z_{N}}\left\Vert \boldsymbol{k}\right\Vert \right).
\]
Let us, as in Appendix \ref{sec:Proof-of-CLT}, write $\ln J_{0}(x)=\sum_{q=1}^{\infty}b_{q}x^{2q}$.
Then
\[
\overline{e^{i\boldsymbol{k}\Delta^{N}(s,h)}}=\exp\left(\sum_{q=1}^{\infty}\sum_{j=\left\lfloor Ns\right\rfloor +1}^{\left\lfloor N(s+h)\right\rfloor }\left\Vert \boldsymbol{k}\right\Vert ^{2q}b_{q}\left(\frac{d_{j}}{\Delta Z_{N}}\right)^{2q}\right).
\]
Define 
\begin{equation}
R_{q}^{N}(h,s):=\frac{\sum_{k=\left\lfloor Ns\right\rfloor +1}^{\left\lfloor N(s+h)\right\rfloor }d_{k}^{2q}}{\left(\sum_{k=1}^{N}d_{k}^{2}\right)^{q}}.\label{eq:RN_hs}
\end{equation}
We now observe that $R_{q}^{N}(h,s)\le R_{q}^{N}$ for all allowed
$s,h$ and all $q>1$. Hence, the assumption that the weights are generic
implies that 
\[
\lim_{N\to\infty}R_{q}^{N}(h,s)=0,\,\,\mathrm{for}\,\,q>1.
\]
Using Eq.~(\ref{eq:Lyapunov_wiener}) and repeating the steps in
Appendix \ref{sec:Proof-of-CLT} to pass the limit inside the exponential
we obtain (remind that $b_{1}=-1/4$) 
\[
\lim_{N\to\infty}\overline{e^{i\boldsymbol{k}\Delta^{N}(s,h)}}=e^{-h\left\Vert \boldsymbol{k}\right\Vert ^{2}/4}.
\]
Fourier transforming,
we obtain the distribution of the random variable
$\Delta^{N}(s,t)$ in the limit $N\to\infty$:
\[
W_{s+h}^{N}-W_{s}^{N}\stackrel{d}{\longrightarrow}\frac{1}{h\pi}e^{-\left\Vert \boldsymbol{x}\right\Vert ^{2}/h}.
\]
Then we obtain
\[
\mathrm{Prob}\left(W:\left|W_{s+h}-W_{s}\right|\le\rho\right)=\frac{2\pi}{h\pi}\int_{0}^{\rho}dr\,re^{-r^{4}/h}.
\]
redefining times such that $s=2s'$ we obtain Eq.~(\ref{eq:conditionb})
of the main text, hence $W_{s'}^{N}$ converges to a Wiener process.

\section{Non-degeneracy and independence}

\label{sec:Non-degeneracy-and-independence}To compute the moments
of the SFF in the literature one often finds a non-degeneracy condition.
Here we show what is the relation between non-degeneracy and linear
independence over the rationals. First we need the following definition:
\begin{defn}
We say that the numbers $\left\{ E_{k}\right\} _{k=1}^{D}$, assumed
to be all different, satisfy the non-degeneracy condition at order
$M\in\mathbb{N}$ ($M$-ND) if, for any pair of sequences $\alpha_{i},\beta_{i}\in\left\{ 1,2,\ldots,D\right\} $,
$i=1,2,\ldots,M$, the equation
\begin{equation}
\sum_{j=1}^{M}E_{\alpha_{j}}=\sum_{j=1}^{M}E_{\beta_{j}}\label{eq:M-ND}
\end{equation}
 implies that $\{\alpha_{1},\alpha_{2},\ldots,\alpha_{M}\}$ is a
permutation of $\{\beta_{1},\beta_{2},\ldots,\beta_{M}\}$, i.e.,~there
exist a permutation $\pi\in\mathbb{S}_{M}$ such that $\alpha_{i}=\beta_{\pi(i)}$. 
\end{defn}
Note that $M$-ND implies $N$-ND for all $1\le N\le M$. It is useful
to go to the occupation number representation. For any $\left\{ \alpha_{i}\right\} _{i=1}^{M}$
we define the (Bosonic) occupation numbers:
\begin{equation}
n_{k}^{\alpha}:=\sum_{i=1}^{M}\delta_{E_{k},E_{\alpha_{i}}},\label{eq:occupation_mapping}
\end{equation}
so $n_{k}^{\alpha}$ counts how many terms in $\left\{ E_{\alpha_{i}}\right\} _{i=1}^{M}$
have energy $E_{k}$. Note that the occupation number representation
is unique modulo permutations, i.e.,~$n_{k}^{\alpha}=n_{k}^{\beta},\,\forall k\,\,\Leftrightarrow\alpha=\pi(\beta)$
for some permutation $\pi$. Moreover, clearly, 
\[
\sum_{j=1}^{M}E_{\alpha_{j}}=\sum_{k=1}^{D}n_{k}^{\alpha}E_{k}.
\]
We have the following result:
\begin{lem*}
The numbers $\left\{ E_{n}\right\} _{n=1}^{D}$ \textup{satisfy non-degeneracy
at order $M$ $\forall M\in\mathbb{N}$ ($\mathbb{N}$ND) if and only
if the energies }$\left\{ E_{n}\right\} _{n=1}^{D}$ are linearly
independent over the rationals.
\end{lem*}
\begin{proof}
First, the $\Leftarrow$ direction. Consider two sequences $\alpha_{i},\beta_{i}$.
Using the mapping Eq.~\eqref{eq:occupation_mapping}, $M$-ND Eq.~(\eqref{eq:M-ND})
can be written as
\begin{equation}
\sum_{k=1}^{D}\left(n_{k}^{\alpha}-n_{k}^{\beta}\right)E_{k}=0.
\end{equation}
 Now independence implies that $n_{k}^{\alpha}=n_{k}^{\beta}$ and
for what we have said this implies $\alpha=\pi(\beta)$ for some permutation
$\pi$. The statement holds for all possible $M$ so this direction
is proven. For the $\Rightarrow$ direction, consider the equation
for linear independence:
\[
\sum_{k=1}^{D}m_{k}E_{k}=0,
\]
with $m_{k}\in\mathbb{Z}$ for $k=1,\ldots,D$. We can always find
sequences $\alpha,\beta$ such that $m_{k}=n_{k}^{\alpha}-n_{k}^{\beta}$.
For example, for $k=1,\ldots,D$, if $m_{k}>0$ we can choose $n_{k}^{\alpha}=m_{k}$
and $n_{k}^{\beta}=0$ and conversely if $m_{k}<0$ we can choose $n_{k}^{\beta}=\left|m_{k}\right|$
and $n_{k}^{\alpha}=0$. By $M$-ND (for some $M$) this implies $m_{k}=n_{k}^{\alpha}-n_{k}^{\beta}=0$,
i.e.,~independence. 
\end{proof}

\section{Moments of the spectral form factor\label{sec:SFF_moments}}

In this section, we give the exact form of the infinite time moments
of the spectral form factor assuming that the energies $\left\{ E_{j}\right\} _{j=1}^{N_{B}}$
satisfy nondegeneracy at order $m$. We recall from Appendix~\ref{sec:Proof-of-commutativity}
\begin{align}
I_{m}:=\left|\chi(t)\right|^{2m} & =\sum_{i_{1}\cdots i_{m}}\sum_{j_{1}\cdots j_{m}}F_{i_{1},j_{1}}\cdots F_{i_{m},j_{m}} \nonumber \\
 & \times\exp\left(-it\sum_{l=1}^{m}\left(E_{i_{l}}-E_{j_{l}}\right)\right)
\end{align}
which we write compactly as
\[
\left|\chi(t)\right|^{2m}=\sum_{s,r}f(s,r)e^{-it(\boldsymbol{E}_{s}-\boldsymbol{E}_{r})}
\]
with multi indices $s=\left(s_{1},\ldots,s_{m}\right)$
corresponding to $\boldsymbol{E}_{s}=E_{s_{1}}+E_{s_{2}}+\cdots E_{s_{m}}$
and similarly for $r$ and $\boldsymbol{E}_{r}$ together with the
coefficients $f(s,r)$. Taking the infinite-time average we get $\overline{e^{-it(\boldsymbol{E}_{s}-\boldsymbol{E}_{r})}}=\delta\left(\boldsymbol{E}_{s}-\boldsymbol{E}_{r}\right)$.
If the energies satisfy $m$ order non-degeneracy, the only way to
fulfill the constraint $\delta\left(\boldsymbol{E}_{s}-\boldsymbol{E}_{r}\right)$
is that the string $r$ is a permutation of $s$. There are
multiplicities, however, that have to be accounted for. The result is 
\begin{align}
\overline{\left|\chi(t)\right|^{2m}} & =\sum_{s,r}f\left(s,r\right)\delta\left(\boldsymbol{E}_{s}-\boldsymbol{E}_{r}\right)\\
 & =\sum_{s}\sum_{\pi\in\mathbb{S}_{m}}f\left(s,\pi\left(s\right)\right)\frac{1}{c\left(s\right)}
\end{align}
the multiplicity $c\left(s\right)$ is the number of different permutations
of $s$ given that some member of $\boldsymbol{E}_{s}$ might be equal.
Having defined, for any string $\left\{ s_{i}\right\} _{i=1}^{m}$
and $k=1,\ldots,N_{B}$, the occupation number $n_{k}^{s}=\sum_{i=1}^{m}\delta_{E_{k},E_{s_{i}}}$
which counts how many energies $E_{k}$ are in $\boldsymbol{E}_{s}$,
clearly one has $\sum_{i=1}^{N_{B}}n_{i}^{s}=m$. Then the multiplicity
factor is given by
\[
c\left(s\right)=\prod_{i=1}^{N_{B}}n_{i}^{s}!
\]
 Now, starting from 
\begin{equation}
\overline{\left|\chi(t)\right|^{2m}}=\sum_{s_{1},\ldots,s_{m}}c\left(s\right)^{-1}\sum_{\pi\in\mathbb{S}_{m}}F_{s_{1},\pi\left(s_{1}\right)}\cdots F_{s_{m},\pi\left(s_{m}\right)},\label{eq:mu_k}
\end{equation}

we use $F_{i,j}=d_{i}d_{j}$ so $f\left(s,\pi\left(s\right)\right)=\prod_{i=1}^{N_{B}}d_{i}^{2n_{i}^{s}}$
and is independent of $\pi$. Then
\begin{eqnarray}
\overline{\left|\chi(t)\right|^{2m}} & = & \sum_{s}\frac{1}{\prod_{i}n_{i}!}\sum_{\pi}\prod_{i=1}^{N_{B}}d_{i}^{2n_{i}^{s}}\\
 & = & \sum_{s}\frac{m!}{\prod_{i}n_{i}!}\prod_{i=1}^{N_{B}}d_{i}^{2n_{i}^{s}}\\
 & = & \sum_{\sum_{i}k_{i}=k}\left[\frac{m!}{\prod_{i}n_{i}!}\right]^{2}\prod_{i=1}^{N_{B}}d_{i}^{2n_{i}}
\end{eqnarray}
where the last line comes from the multinomial formula that allows
to go from index labels $s=\left(s_{1},\ldots,s_{k}\right)$ to occupation
number labels $n_{i},\,i=1,\ldots,N_{B}$. Finally the $m$-th moment
of the spectral form factor is given by 
\begin{align}
\overline{\left|\chi(t)\right|^{2m}} & =\left(m!\right)^{2}\sum_{\sum_{i}n_{i}=m}\prod_{i=1}^{N_{B}}\left(\frac{d_{i}^{n_{i}}}{n_{i}!}\right)^{2}\label{eq:I_2m}\\
 & =\sum_{\sum_{i}n_{i}=m}\left(\begin{array}{c}
m\\
n_{1},\cdots,n_{N_{B}}
\end{array}\right)^{2}\prod_{i=1}^{N_{B}}d_{i}^{2n_{i}},
\end{align}
where in the last line we introduced the multinomial coefficient.
In case $d_{j}=1$ for all $j$ the above reduces to the expression
found in \citep{borwein_arithmetic_2011}. Equation (\ref{eq:I_2m}),
although exact, is difficult to manage. For example, it is not clear how
to estimate the order of magnitude or the leading contribution to
the moments. In order to obtain a more manageable expression let us
define the following generating function
\begin{align}
B\left(z\right) & :=\sum_{k=0}^{\infty}\frac{I_{k}}{\left(k!\right)^{2}}z^{k}\label{eq:B-def}\\
 & =\overline{I_{0}\left(2\left|\chi(t)\right|\sqrt{z}\right),}
\end{align}
where $I_{0}$ is the modified Bessel function of the first kind.
The above sum can be obtain in closed form, indeed,
\begin{eqnarray}
B\left(z\right) & = & \sum_{k=0}^{\infty}\sum_{\sum_{i}k_{i}=k}\prod_{i=1}^{N_{B}}\frac{\left(d_{i}^{2}z\right)^{k_{i}}}{\left(k_{i}!\right)^{2}}\\
 & = & \prod_{i=1}^{N_{B}}\sum_{n=0}^{\infty}\frac{\left(d_{i}^{2}z\right)^{n}}{\left(n!\right)^{2}}.
\end{eqnarray}
The last sum is again $I_{0}$, more precisely
\begin{equation}
B\left(z\right)=\prod_{i=1}^{N_{B}}I_{0}\left(2d_{i}\sqrt{z}\right) \, ,  \label{eq:B-res}
\end{equation}
In principle one can obtain the moments via
\[
I_{m}=m!\left.\frac{d^{m}B}{dz^{m}}\right|_{z=0},
\]
but we are proposing a better expression momentarily. 

For convenience we define $I\left(x\right)=I_{0}\left(2\sqrt{x}\right)$, and the coefficients $a_{n}$ via the following series ($I(x)\ge1$
for $x\ge0$ so the logarithm is well defined and analytic in a neighborhood of zero) 
\[
\ln I\left(x\right)=\sum_{n=0}^{\infty}a_{n}\frac{x^{n}}{n!}.
\]
then 
\begin{eqnarray}
\ln\left(B\left(z\right)\right) & = & \sum_{i=1}^{N_B}\ln I\left(d_{i}^{2}z\right)\nonumber \\
 & = & \sum_{i=1}^{N_B}\sum_{n=0}^{\infty}a_{n}d_{i}^{2n}\frac{z^{n}}{n!}\nonumber \\
 & = & \sum_{n=0}^{\infty}a_{n}X_{n}\frac{z^{n}}{n!},\label{eq:B_LE}
\end{eqnarray}
where we defined $X_{n}=\sum_{i=1}^{N_{B}}d_{i}^{2n}$. The coefficients
$a_{n}$ are the cumulants of a random variable with moments $\mu_{k}=1/k!$.
So $a_{1}=1$, $a_{2}=1/2-1^{2}=-1/2$, $a_{3}=\ldots=2/3$ and so
on. Hence the moments $I_{k}/k!$ have cumulants $a_{n}X_{n}$. We
can then use the standard recursion between moments and cumulants:
\[
\frac{I_{p}}{p!}=\sum_{q=1}^{p-1}\left(\begin{array}{c}
p-1\\
q-1
\end{array}\right)a_{q}X_{q}\frac{I_{p-q}}{(p-q)!}+a_{p}X_{p}.
\]
 which is Eq.~(\ref{eq:recursion}) from the main text. With this
recursion one easily obtains Eqs.~(\ref{eq:moments1}-\ref{eq:moments2})
of the main text and also, for example, assuming $5$-ND,
\begin{align}
I_{4} & =24\left(X_{1}\right)^{4}-72\left(X_{1}\right)^{2}X_{2}+18\left(X_{2}\right)^{2}\nonumber \\
 & +64X_{1}X_{3}-33X_{4}\\
I_{5} & =120\left(X_{1}\right)^{5}-600X_{2}\left(X_{1}\right)^{3}+800X_{3}\left(X_{1}\right)^{2}\nonumber \\
 & +450\left(X_{2}\right)^{2}X_{1}-825X_{4}X_{1}\nonumber \\
 & -400X_{2}X_{3}+456X_{5}\,.
\end{align}

\section{Limit distribution for free Fermionic spectra\label{sec:distribution_free}}

Consider $H$ with spectrum $E_{\boldsymbol{n}}=\sum_{k=1}^{L}n_{k}\Lambda_{k}+E_{0}$
with $n_{k}=0,1$. We first compute ($z=\beta+it$)
\begin{align*}
\tr\left(e^{-zH}\right) & =e^{-zE_{0}}\sum_{\left\{ n_{k}\right\} }\exp\left(-z\sum_{k=1}^{L}n_{k}\Lambda_{k}\right)\\
 & =e^{-zE_{0}}\sum_{\left\{ n_{k}\right\} }e^{-zn_{1}\Lambda_{1}}\cdots e^{-zn_{L}\Lambda_{L}}\\
 & =e^{-zE_{0}}\prod_{k=1}^{L}\left(1+e^{-z\Lambda_{k}}\right).
\end{align*}
From here on we stick to infinite temperature i.e.,~$z=it$. If some
levels are degenerate we get
\begin{align}
\left|\tr\left(e^{-itH}\right)\right|^{2} & =\prod_{j=1}^{L_{B}}\left|1+e^{-it\epsilon_{j}}\right|^{2g_{j}} \nonumber \\
 & =\prod_{j=1}^{L_{B}}\left(2\cos\left(t\epsilon_{j}/2\right)\right)^{2g_{j}}.
\end{align}
Assuming that the levels $\left\{ \epsilon_{j}\right\} _{j=1}^{L_{B}}$are
independent implies that (upon considering $t$ a random variable)
$\left|\tr\left(e^{-itH}\right)\right|^{2}$ is a product of independent
random variables $Y_{j}$. So the logarithm of the product is a sum
of independent random variables $\ln Y_{j}$. The moment generating
function of $\ln Y_{j}$ is

\begin{align}
\overline{\left(Y_{j}\right)^{\lambda}} & =\int_{0}^{2\pi}\frac{d\vartheta}{2\pi}\left(2\cos\left(\vartheta\right)\right)^{2g_{j}\lambda}=\frac{4^{\lambda g_{j}}\Gamma\left(\frac{1}{2}+\lambda g_{j}\right)}{\sqrt{\pi}\Gamma\left(1+\lambda g_{j}\right)}\label{eq:frac_moments}\\
 & =G(\lambda g_{j})=\exp\left[\sum_{n=2}^{\infty}c_{n}\left(g_{j}\lambda\right)^{n}\right],
\end{align}
where in the last equation we defined the function $G$ and the
coefficients $c_{n}$ (since the function in (\ref{eq:frac_moments})
is analytic and positive for $\lambda>-1/2$ its logarithm is analytic
in the same region). Explicitly $c_{2}=\pi^{2}/6$. 

Let us define the following random variable
\[
Q^{L}=\frac{\ln\left|\chi(t)\right|^{2}}{\sqrt{\sum_{j=1}^{L_{B}}g_{j}^{2}}},
\]
and the ratios 
\[
S_{p}^{L_{B}}=\frac{\sum_{j=1}^{L_{B}}\left(g_{j}\right)^{2}}{\left(\sum_{j=1}^{L_{B}}g_{j}^{2}\right)^{p/2}}.
\]
By the same arguments as in Appendix \ref{sec:Proof-of-CLT} the assumption
that the one-particle degeneracies $g_{j}$  satisfy the Lyapunov condition Eq.~\eqref{eq:generic_integrable} implies that
\[
\lim_{L\to\infty}S_{p}^{L}=0,\,\,\forall p>2.
\]
The random variable $Q^{L}$ has the following characteristic function
\[
\overline{e^{ikQ^{L}}}=\exp\sum_{j=1}^{L_{B}}\sum_{p=2}^{\infty}c_{p}S_{p}^{L_{B}}\left(ik\right)^{p}.
\]

To show that we can pass the limit inside the exponential note that
$\left|S_{p}^{L_{B}}\right|\le1$ for all $p$ and since $\ln(G(\lambda))$
is analytic, the coefficients are bounded by $\left|c_{p}\right|\le C^{p+1}$
for some constant $C$ and $\sum_{p=2}^{\infty}C^{p+1}\left|k\right|^{p}<\infty$
for $\left|k\right|$ small enough so that the limit and the sum can
be swapped by Tannery's theorem. Finally, we obtain, in the thermodynamic
limit, 
\[
\lim_{L\to\infty}\overline{e^{i\lambda Q^{L}}}=e^{-(\pi^{2}/6)\lambda^{2}}.
\]
After Fourier transform, we get
\[
\frac{\ln\left|\chi(t)\right|^{2}}{\sqrt{\sum_{j=1}^{L_{B}}g_{j}^{2}}}\stackrel{d}{\longrightarrow}\mathsf{N}\left(0,\frac{\pi^{2}}{3}\right).
\]

For finite temperature, $\beta \neq 0$, the steps are the same, but we could not evaluate the variance in closed form, which now depends on $\beta$. 

\section{Moments for free Fermionic spectra\label{sec:Moments_free}}

Here we want to compute $I_{M}=\overline{\left|\chi(t)\right|^{2M}}$
when $H$ has the free Fermionic spectrum, under the assumption that
the $\left\{ \epsilon_{j}\right\} _{j=1}^{L_{B}}$ are independent.
Following the previous section we have
\begin{align*}
\left|\chi(t)\right|^{2M} & =\prod_{j=1}^{L_{B}}\left(e^{it\epsilon_{j}/2}+e^{-it\epsilon_{j}/2}\right)^{2g_{j}M}\\
 & =\prod_{j=1}^{L_{B}}\sum_{k=0}^{2g_{j}M}\left(\begin{array}{c}
2g_{j}M\\
k
\end{array}\right)e^{-it\epsilon_{j}(g_{j}M-k)}\\
 & =\sum_{k_{1}=0}^{2g_{j}M}\cdots\sum_{k_{L_{B}}=0}^{2g_{j}M}\left(\begin{array}{c}
2g_{j}M\\
k_{1}
\end{array}\right)\cdots\left(\begin{array}{c}
2g_{j}M\\
k_{L_{B}}
\end{array}\right)\\
 & \times\exp\left(-it\sum_{j=1}^{L_{B}}\epsilon_{j}(g_{j}M-k_{j})\right).
\end{align*}
Taking the infinite time average independence of the $\epsilon_{j}$
implies that only the terms with $k_{j}=g_{j}M$ remain and we get
\[
\overline{\left|\chi(t)\right|^{2M}}=\prod_{j=1}^{L_{B}}\left(\begin{array}{c}
2g_{j}M\\
g_{j}M
\end{array}\right).
\]

\section{Numerics on physical models}

\subsection{Check of the Lyapunov condition in Eq.(\ref{eq:generic_weights}) at high temperature}
\label{sec:numerics} 

We show that at infinite temperature $T = \infty$ (with $\beta = 1/T$) implying $\rho=\mathbb{I} $,  in case of physical models (local Hamiltonians), with spectral resolution $H=\sum_{j=1}^{N_B} E_j \Pi_j $, the weights $d_j =\tr (\Pi_j)$, satisfy the Lyapunov condition Eq.~(\ref{eq:generic_weights}), i.e.,~are {\it generic}. For a given Hamiltonian size, we define
\begin{equation}
R_{q}^{n}=\frac{\sum_{j=1}^{n}d_{j}^{2q}}{\left(\sum_{j=1}^{n}d_{j}^{2}\right)^{q}},\label{eq:RN}
\end{equation}
where $n\in\{1,2,\ldots.N_{B}\}$. We investigate this behavior using the following XXZ Hamiltonian with next-to-nearest-neighbor (NNN) interactions (Eq.~\eqref{XXZNNNHamiltonian} of the main text) with periodic boundary conditions, $\sigma_{L+j}^\alpha = \sigma_j^\alpha$: 
\begin{align}
H & =\sum_{j=1}^{L}\left(\sigma_{j}^{x}\sigma_{j+1}^{x}+\sigma_{j}^{y}\sigma_{j+1}^{y}+\Delta\sigma_{j}^{z}\sigma_{j+1}^{z}\right) \nonumber \\
 & +\alpha\sum_{j=1}^{L}\left(\sigma_{j}^{x}\sigma_{j+2}^{x}+\sigma_{j}^{y}\sigma_{j+2}^{y}+\Delta\sigma_{j}^{z}\sigma_{j+2}^{z}\right). 
\end{align}

This model interpolates between three regimes:
\begin{itemize}
    \item \textit{non-integrable} case ($\alpha \neq 0$),
    \item \textit{Bethe Ansatz (BA) integrable} case ($\alpha = 0$),
    \item \textit{quadratic quasi-free integrable} case ($\Delta = 0$, $\alpha = 0$).
\end{itemize}
We numerically diagonalize the Hamiltonian for finite system sizes $N$ and identify eigenvalue degeneracies using a numerical threshold of $\epsilon = 10^{-8}$.

In Fig.~\ref{fig:RN_check}, we present numerical evidence that the weights $d_j$ behave generically in both the integrable (Bethe Ansatz) and non-integrable regimes, i.e., the ratio decays exponentially with the system size $N$ (upper panels), and as a power-law with the eigenvalue index $n$ (lower panels). Note that, as shown in the proof of Theorem 2,  if, for $M\to\infty$, $R_{q}^{M}\to0$ for a certain value of $q>1$, then $R_{q}^{M}\to0$ for all $q>1$, and so it is sufficient to check the statement for a single $q$.

\begin{figure}[t!]
\begin{centering}
\includegraphics[width=8.5cm]{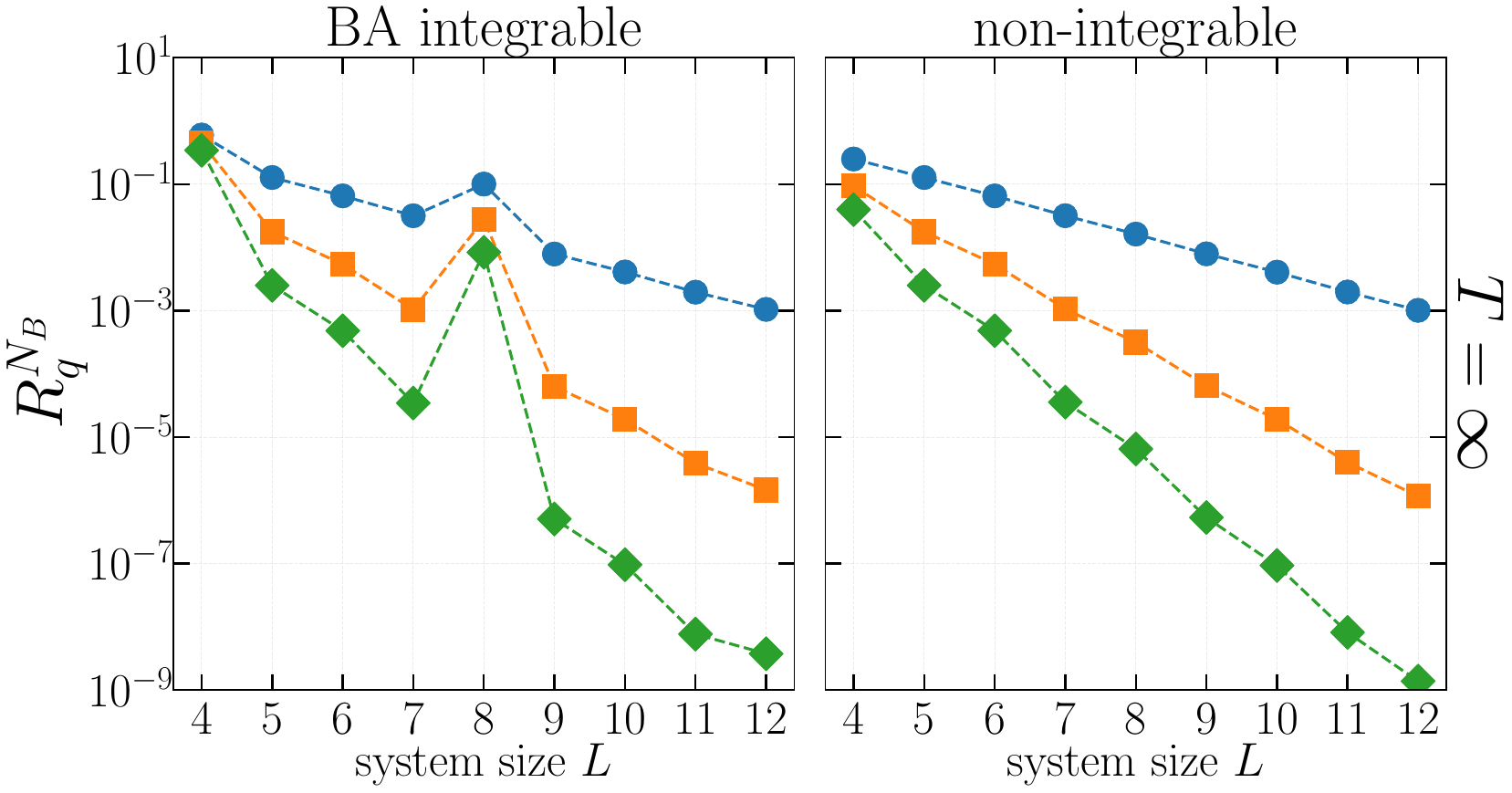}
\par\end{centering}
\begin{centering}
\includegraphics[width=8.5cm]{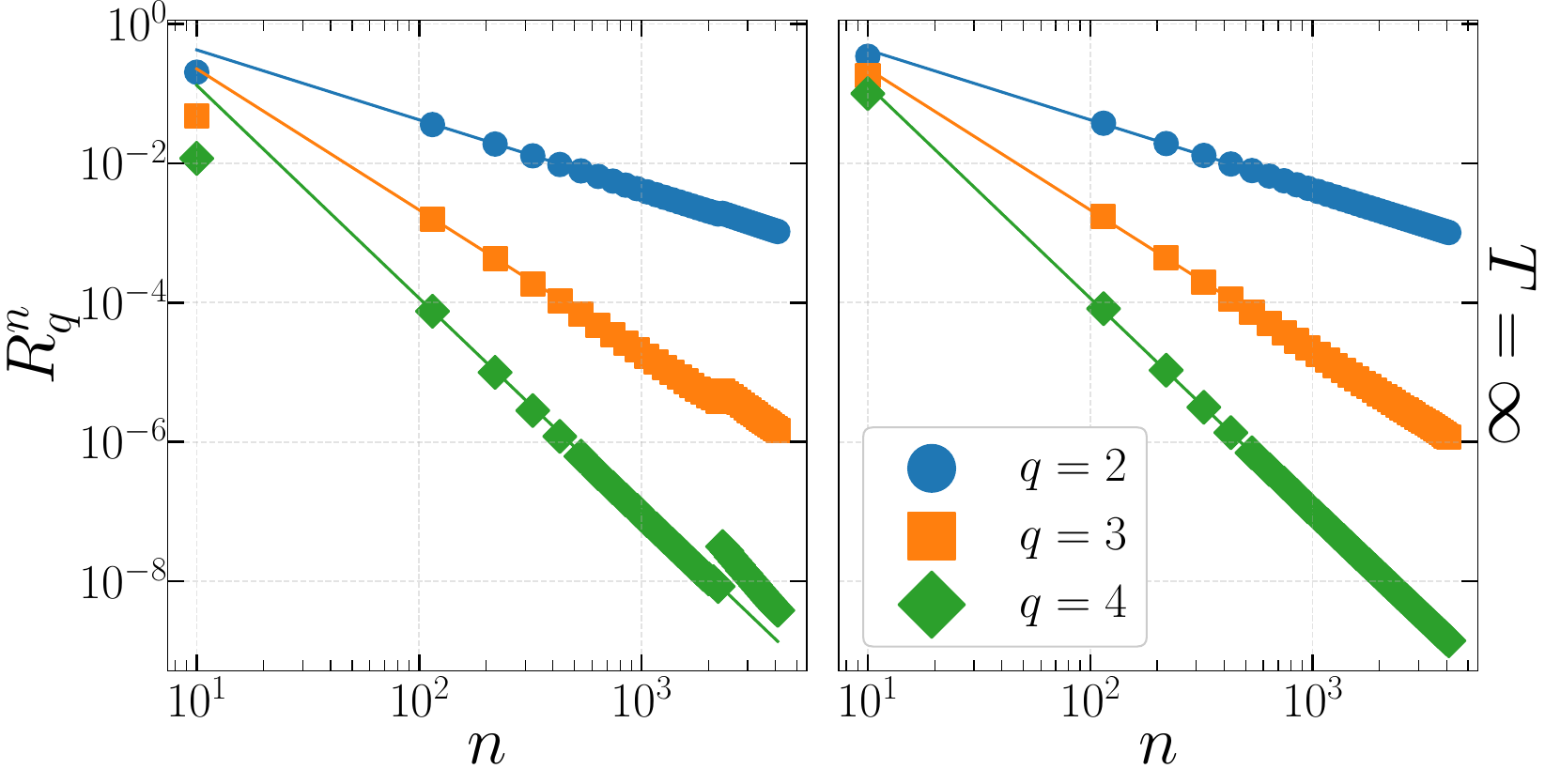}
\par\end{centering}
\caption{Numerical check of~\eqref{eq:generic_weights} for the paradigmatic example of a physical (local) Hamiltonian given by XXZ+NNN chain in~\eqref{XXZNNNHamiltonian} at infinite temperature.  \textit{Upper panels}: Decay of the ratios $R_q^{N_B}$ for various system sizes $L$, where $N_B$ denotes the total number of blocks in the Hamiltonians (corresponding to the number of $d_j$'s). Dashed lines are included as guides for the eye. The bump appearing at $L=8$ (and to a lesser extent at $L=4$) is  due to the fact that for $L=2^{\mathrm{integer}}$ the degeneracies are larger and this gives rise to an increase in $R^{N_B}_q$. 
\textit{Lower panels}: Decay of the ratios $R_q^{n}$ for fixed system size $L=12$, the largest size for which we computed the full spectrum. Solid lines are fits to the power-law function $f(n) = b/n^{a-1}$, with the extracted exponent satisfying $a \approx q$ up to small deviations. Panels on the left correspond to the Bethe Ansatz (BA) integrable case with $\alpha=0$ and  $\Delta=0.1$, while panels on the right show the non-integrable (NI) case with $\alpha=0.1$ and $\Delta=0.5$. The degeneracies on the left (BA) belong to the set $\{1,2,3,4,6,10\}$ while those on the right (NI) to $\{1,2,3,4\}$. The bump on the left panel corresponds to the first occurrence of $d_j=10$ and is understandably more pronounced for large $q$. 
 \label{fig:RN_check}
}

\end{figure}

\begin{figure}[ht!]
\begin{centering}
\includegraphics[width=8.5cm]{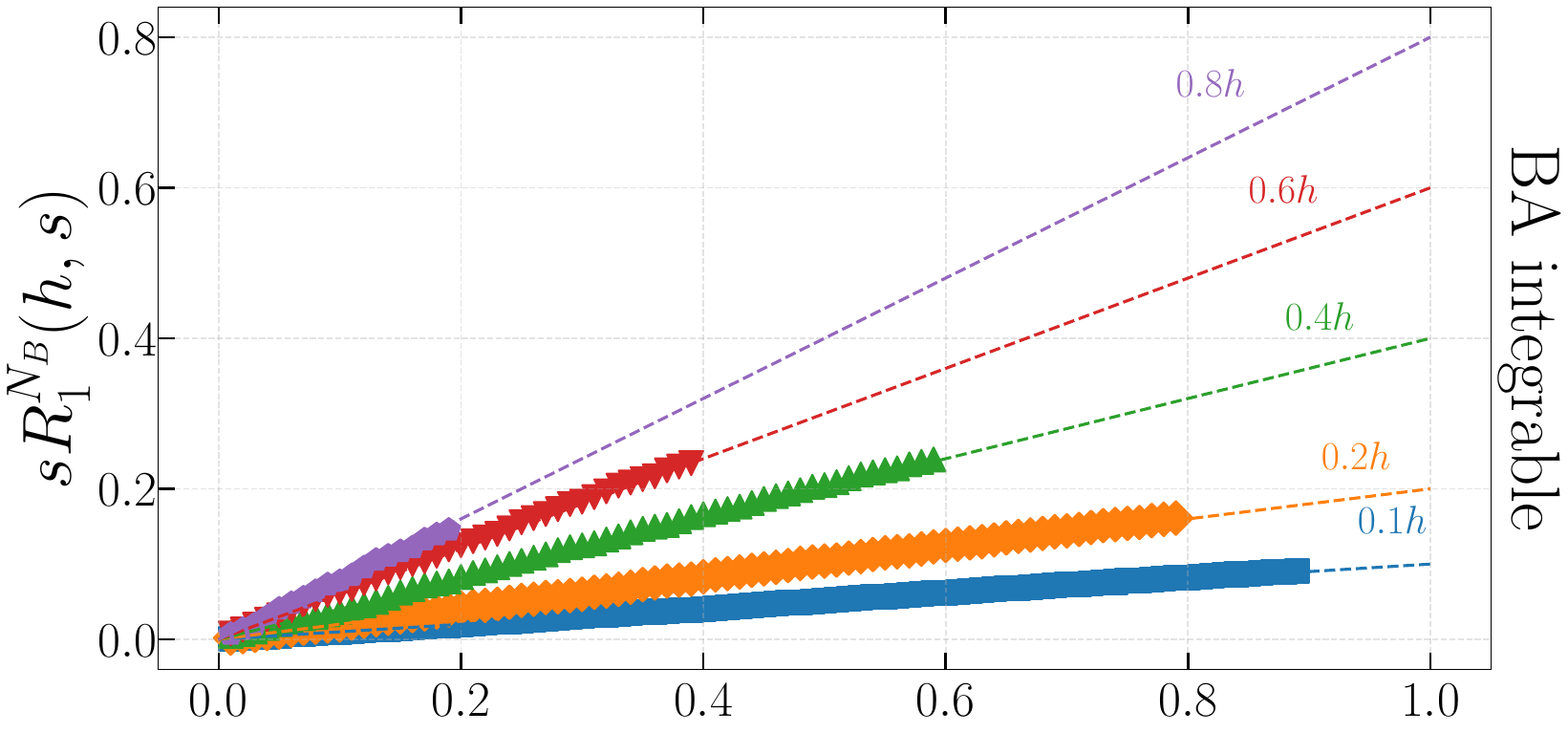}
\par\end{centering}
\begin{centering}
\includegraphics[width=8.5cm]{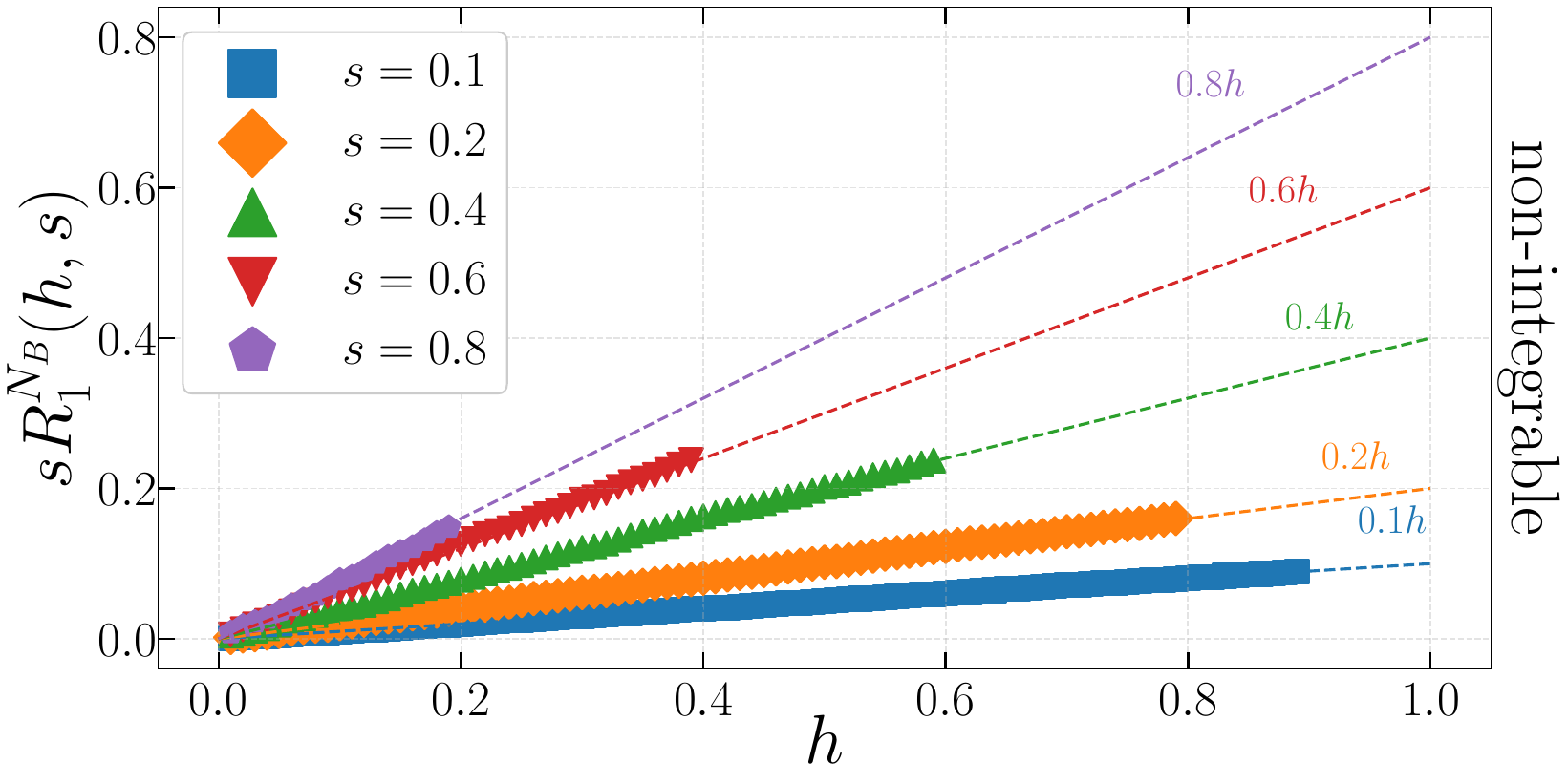}
\par\end{centering}
\caption{Numerical check of~\eqref{eq:Lyapunov_wiener} for the paradigmatic example of a physical (local) Hamiltonian given by XXZ+NNN chain in~\eqref{XXZNNNHamiltonian} at infinite temperature $T=\infty$. We fix the system size
$L=12$, show the result of the quantity $sR_{1}^{N}(h,s)$ which is expected to be $sh$. \textit{Upper panel}: Bethe Ansatz integrable case with $\alpha=0$ and  $\Delta=0.1$. \textit{Lower panel}: non-integrable case with parameters $\alpha=0.1$ and $\Delta=0.1$. } \label{fig:RN_hs_check}
\end{figure}

%


\subsection{Check of Eq.~(\ref{eq:Lyapunov_wiener}) at high temperature}
We now turn to verifying condition~\eqref{eq:Lyapunov_wiener} from the main text. By Theorem~3, this condition, together with the assumption of independence of the energy levels, guarantees that the associated random walk converges to a Wiener process in the thermodynamic limit. To recall, the condition requires that
\begin{equation}
    \lim_{N \to \infty} R_1^N(h, s) = h \, , \label{eq:check-wiener}
\end{equation}
where we recall Eq.~(\ref{eq:RN_hs}):
\begin{equation}
R_{q}^{N}(h,s):=\frac{\sum_{k=\left\lfloor Ns\right\rfloor +1}^{\left\lfloor N(s+h)\right\rfloor }d_{k}^{2q}}{\left(\sum_{k=1}^{N}d_{k}^{2}\right)^{q}}.
\end{equation}
In Fig.~\ref{fig:RN_hs_check}, we numerically plot $s R_1^N(h, s)$, which should asymptotically approach $s h$ for large system sizes $N$, for various values of $h$ and $s$. The results show excellent agreement with the theoretical prediction (dotted lines), confirming the expected behavior.

In summary, for the XXZ+NNN chain, we have verified that both Lyapunov-type conditions for the spectral weights---Eqs.~(\ref{eq:generic_weights}) and~(\ref{eq:Lyapunov_wiener})---are satisfied. These are necessary prerequisites for Theorems~\ref{thm:CLT_lyapunovplus} and~\ref{thm:Wiener} to apply. What remains is the verification of linear independence of the energy spectrum over the rationals---a requirement that is, in practice, numerically intractable. Nevertheless, for non-integrable systems, spectral independence is generally expected. In contrast, the Bethe-Ansatz integrable case is more subtle: it may yield spectra that are either rationally independent or contain many rational relations.

\subsection{Violation of the Lyapunov condition Eq.~(\ref{eq:generic_weights}) at low temperature}

In Fig.~\ref{fig:finitetemp}, we demonstrate the breakdown of the Lyapunov condition Eq.~(\ref{eq:generic_weights}) at low temperature, thus signaling the inapplicability theorems~2 and 3. In particular, with the decrease of the temperature parameter $T$, we observe (from the top to bottom panels) that the ratio $R_q^{N_B}$ changes its exponential decay with the system size and flattens off as the zero temperature limit is approached. 

\begin{figure}[t!]
\begin{centering}
\includegraphics[width=8cm]{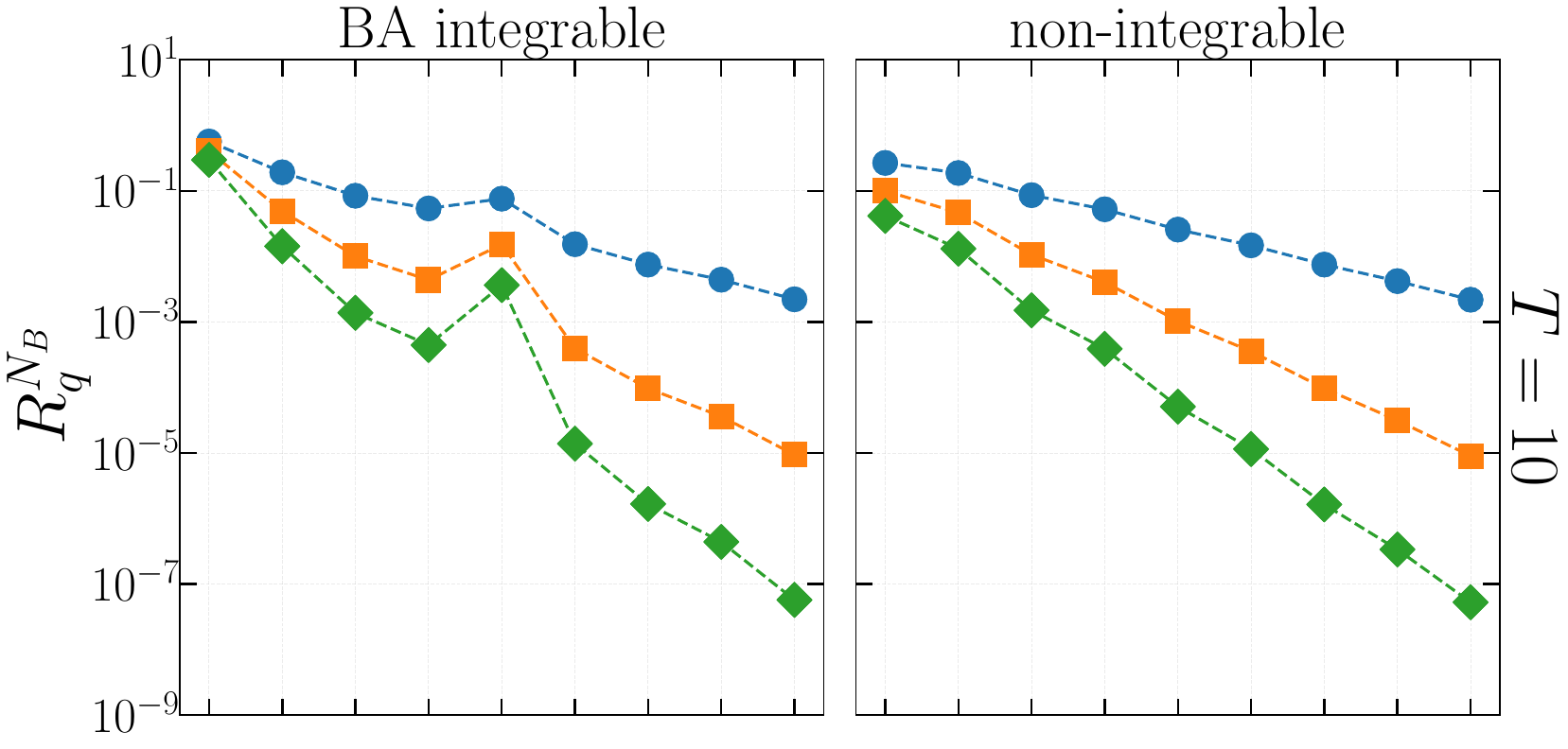}
\par\end{centering}
\begin{centering}
\includegraphics[width=8cm]{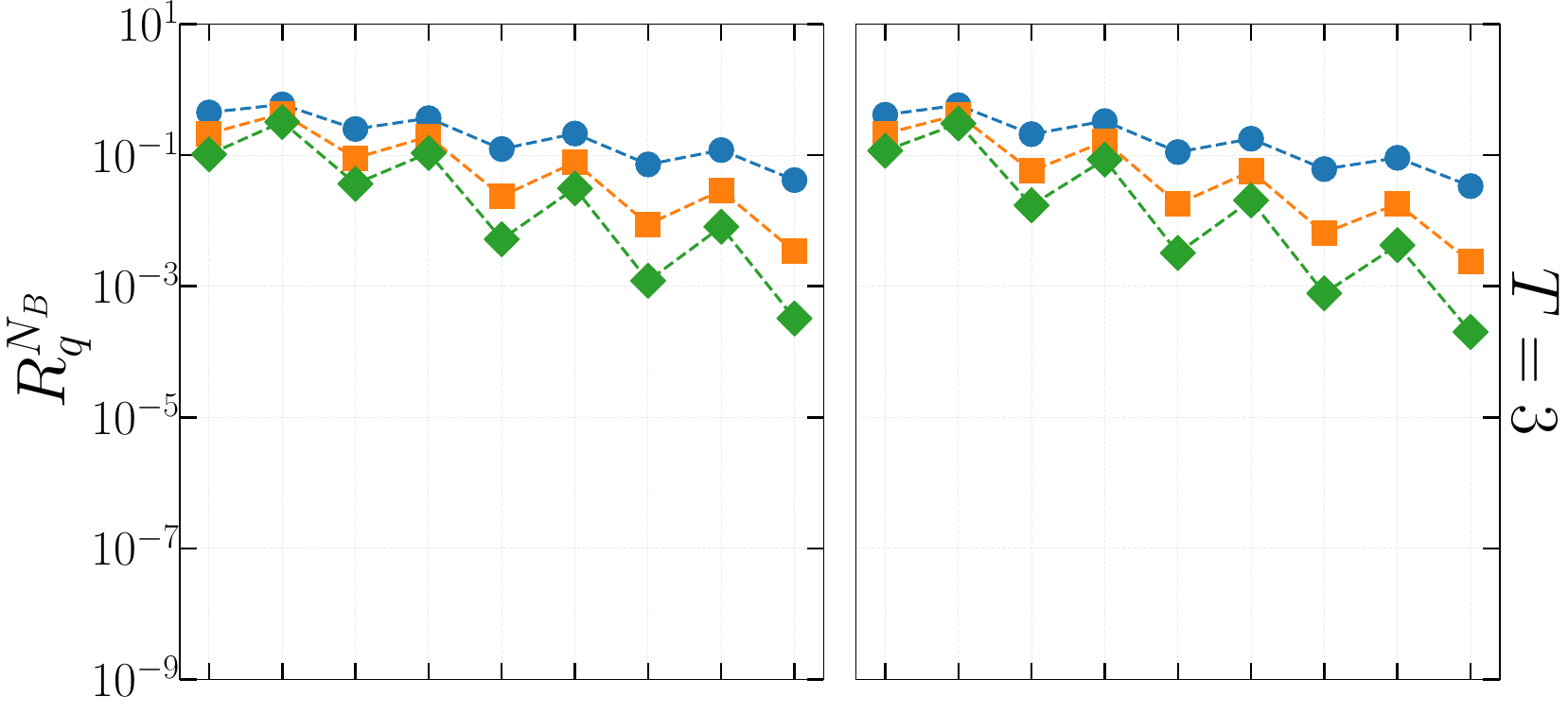}
\par\end{centering}
\begin{centering}
\includegraphics[width=8cm]{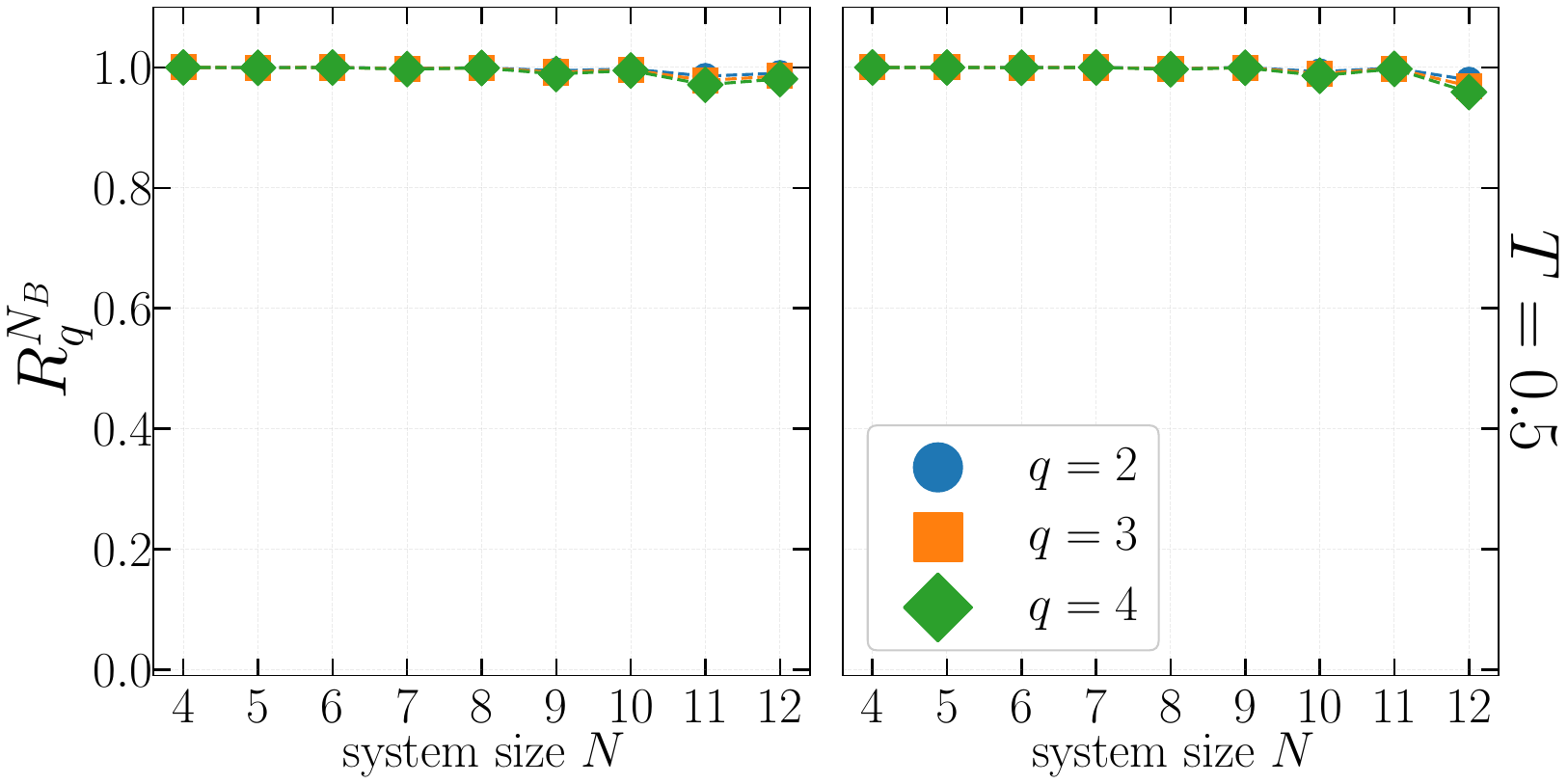}
\par\end{centering}
\caption{ Numerical check of~\eqref{eq:generic_weights} for the paradigmatic example of a physical (local) Hamiltonian given by XXZ+NNN chain in~\eqref{XXZNNNHamiltonian} at low and different temperatures $T$ and breakdown of the Lyapunov condition. Decay of the ratios $R_q^{N_B}$ for various system sizes $L$, where $N_B$ denotes the total number of blocks in the Hamiltonians (corresponding to the number of $d_j$'s). Dashed lines are included as guides for the eye. Significant deviation from an exponential decay of the ratios signals the breakdown of the Lyapunov condition Eq.~\eqref{eq:generic_weights}. Decrease of the temperature  from the top to the bottom panel shows how exactly the condition is not longer satisfied in the low temperature limit. }\label{fig:finitetemp}
\end{figure}

\begin{figure}[ht!]
\begin{centering}
\includegraphics[width=8.5cm]{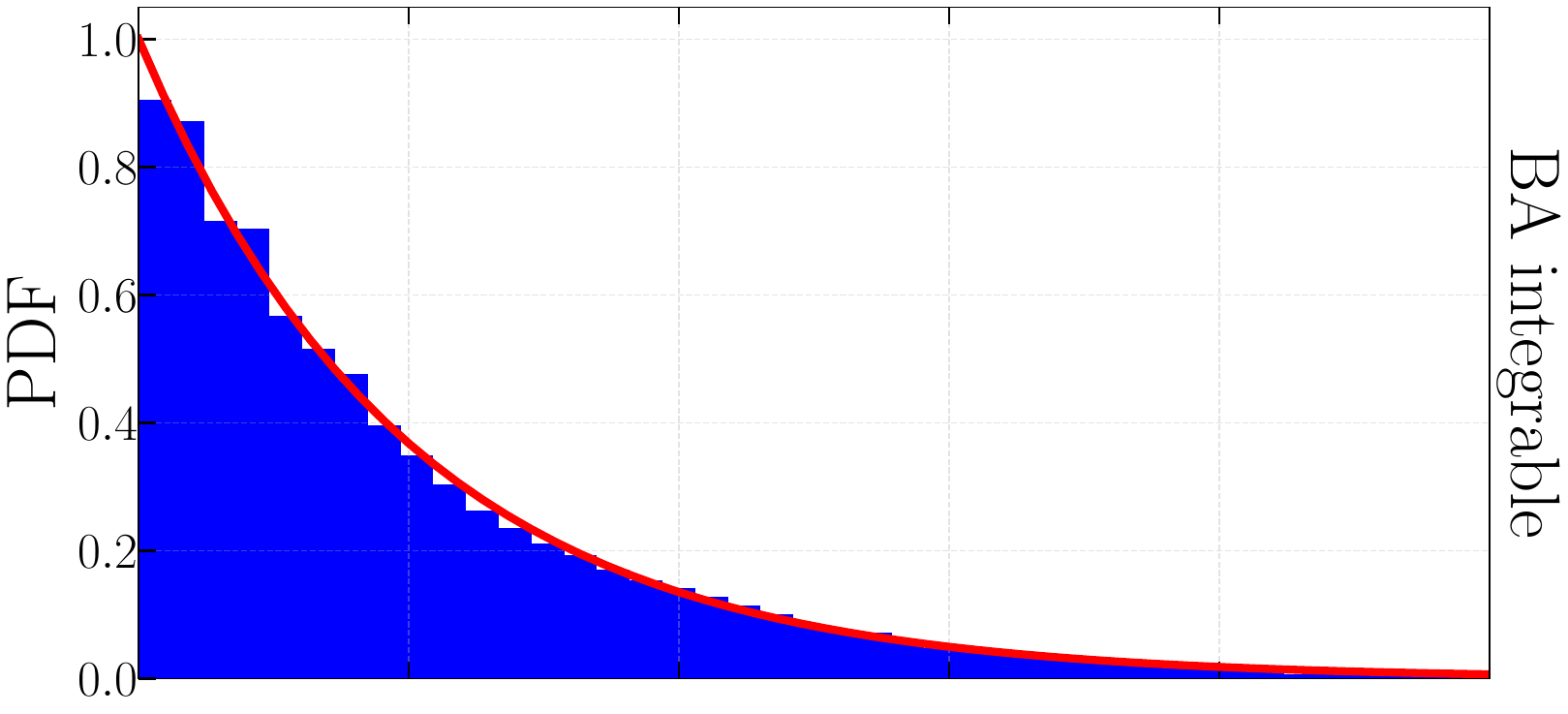}
\par\end{centering}
\begin{centering}
\includegraphics[width=8.5cm]{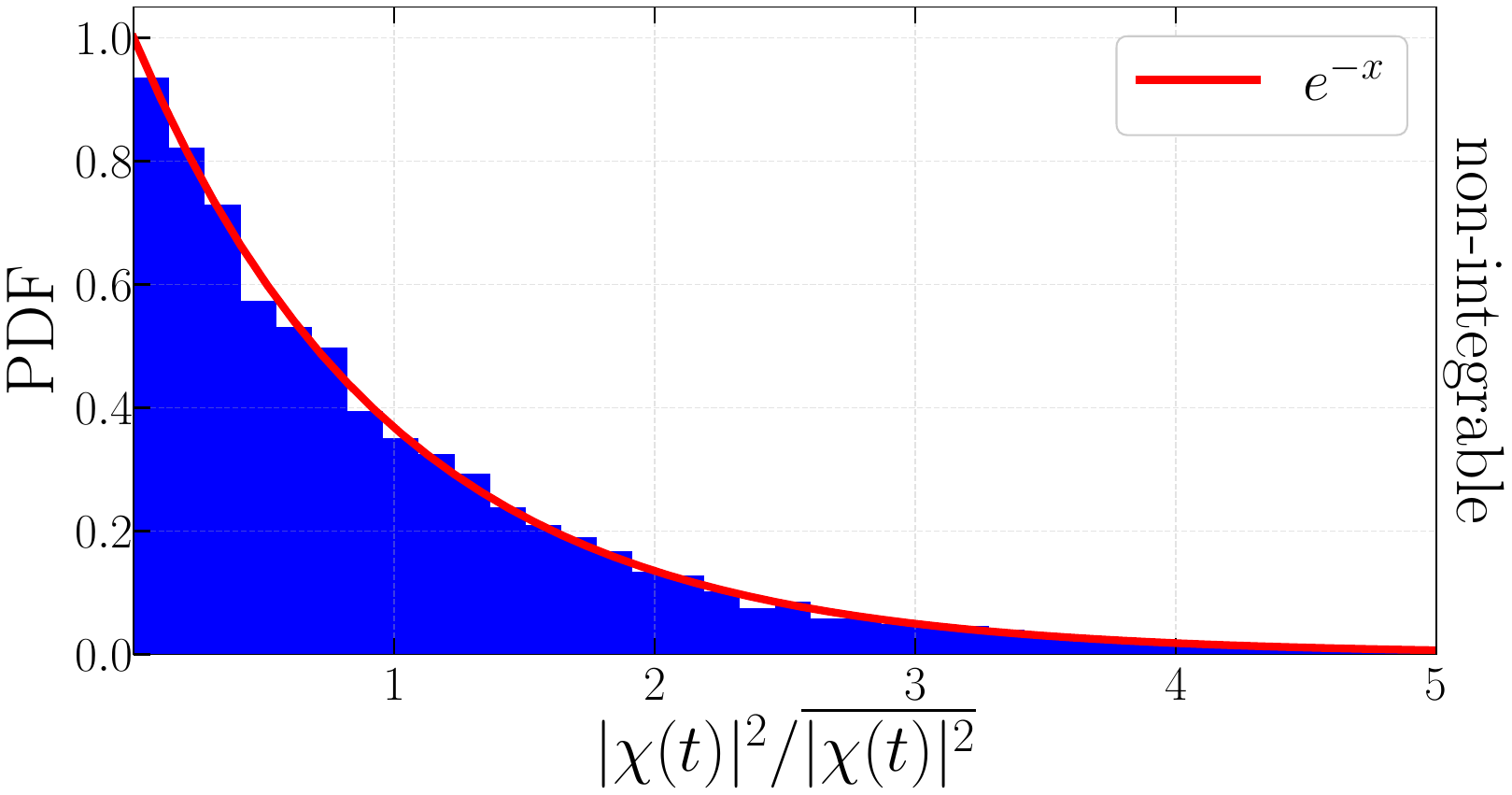}
\par\end{centering}
\caption{Probability distribution of the normalized Spectral Form Factor (SFF) for the XXZ + NNN chain spectrum at infinite temperature ($\rho=\1$) for system size $L = 8$. The distribution is obtained sampling in the time domain $[10^5,2 \cdot 10^5]$.  The prediction of Theorem~\ref{thm:CLT_lyapunovplus}, $e^{-x}$, is shown in red. \textit{Upper panel}: non-integrable spectrum with $\Delta = 0.1$ and $  \alpha = 0.1$. \textit{Lower panel}: Bethe Ansatz integrable spectrum with $\Delta = 0.1$ and $  \alpha = 0.0$.  
Notice the small system sizes used to perform the checks, and how well the exponential function captures the behavior.  
} \label{fig:firstmomentsXXZ}
\end{figure}

\subsection{Probability distribution of the spectral form factor $\vert \chi (t) \vert^2$} \label{appedix:numericalcheck}
\label{sub:pdf_SFF}
 In Fig.~\ref{fig:firstmomentsXXZ}, we show the full probability distribution of the normalized SFF, $\vert \chi (t) \vert^2/ \overline{\vert \chi (t) \vert^2}$ both in the non-integrable and Bethe-Ansatz integrable phases of the XXZ+NNN chain (right panels). It is apparent that in both the Bethe-Ansatz and in the non-inegrable phases, the behavior of the SFF is the one predicted by Theorem 2: $\vert \chi (t) \vert^2/ \overline{\vert \chi (t) \vert^2} \to \mathsf{Exp}(1)$ for increasing system size. 
 It is interesting to notice, as shown in Table~\ref{tab:Expectation-of-different}, that the BA integrable regime cannot be differentiated from fully non-integrable theories by the long-time behavior of the SFF.

In Fig.~\ref{fig:free-SFF} we  demonstrate that quasi-free theories instead,  exhibit a  behavior of the SFF predicted by Theorem 5, i.e.,~the SFF is lognormal and its properly normalized logarithm becomes Gaussian in the thermodynamic limit. More precisely, using the notation of Theorem 5, 
\[
\frac{\log\left(\left|\chi(t)\right|^{2}\right)}{\sqrt{\sum_{j=1}^{L_{B}}g_{j}^{2}}}\to\mathsf{N}\left(0,\frac{\pi^{2}}{3}\right). \label{eq:log-gaussian}
\]

To check Eq.~(\ref{eq:log-gaussian}) we perform numerics on the XY model given by the following Hamiltonian
\begin{equation}
    H_{XY}=-\sum_{i=1}^{L}\left[\left(\frac{1+\gamma}{2}\right)\sigma_{i}^{x}\sigma_{i+1}^{x}+\left(\frac{1-\gamma}{2}\right)\sigma_{i}^{y}\sigma_{i+1}^{y}+h\sigma_{i}^{z}\right]. \label{eq:XY_model}
\end{equation}
We use periodic boundary conditions on the fermions  (see e.g.~\cite{campos_venuti_equivalence_2010} for a discussion). Standard diagonalization brings the model to 
\[
H_{XY}=\sum_{k\in BZ}\Lambda_{k}\left(\eta_{k}^{\dagger}\eta_{k}-\frac{1}{2}\right),
\]
where the Brillouin zone is $BZ=\{k| k=2\pi n/L, \,n=0,1,\ldots,L-1 \} $ and the single particle dispersion reads 
\[
\Lambda_k = 2 \sqrt{\gamma^{2}+h^{2}+\left(1-\gamma^{2}\right)\cos\left(k\right)^{2}+2 h\cos\left(k\right)}, 
\]
so that the SFF reads (see the proof of Theorem 5)
\begin{align}
|\chi(t)|^2  & =\prod_{k=1}^{L}\left|1+e^{-it\Lambda_{k}}\right|^{2} \nonumber \\
 & =\prod_{j=1}^{L_{B}}\left(2\cos\left(t\epsilon_{j}/2\right)\right)^{2g_{j}}.
\end{align}

For general values of $h,\gamma$, the one particle spectrum $\Lambda_k$ is doubly degenerate so that $g_j=2, \, \forall j$, and $L_B=L/2$. At the point $h=\gamma=0$ the spectrum becomes $\Lambda_k = 2 \vert \cos(k) \vert$. In this case, one can prove that $\epsilon_k$ is independent over the rationals for $L$ a prime number, while it is expected not to be independent when $L$ is not prime \cite{campos_venuti_exact_2011}.  
For other values of  $h,\gamma$ we observe the behavior expected from Theorem 5, irrespective if $L$ is prime or not, indicating that the one particle spectrum is likely independent, see Fig.~\ref{fig:free-SFF} upper panels. For $h=\gamma=0$ the largest discrepancy from Eq.~(\ref{eq:log-gaussian}) is expected when $L$ is highly composite. So in the lower panels of Fig.~\ref{fig:free-SFF}  we show results for $L=128=2^7$ which shows a large discrepancy (left) with those of $L=127$ which is prime (right). See also \cite{Riddell_Bertini_2024} for analogous discussions on free Fermionic spectra. 
Note that, for free theories the large parameter in the CLT is $L=\log_2 (D)$, and the errors from applying the CLT at finite size are of the order of $O(L^{-1/2})$ as opposed to $O(D^{-1/2})$ for the non-integrable case. As a consequence, to check Theorem \ref{thm:CLT_lyapunovplus} for free models, we need to go to considerably larger sizes, while to check Theorem \ref{thm:free_fermion} in non-integrable models sizes of order of $L=8$ suffice.

\begin{figure}[t!]
\begin{centering}
\includegraphics[width=4cm]{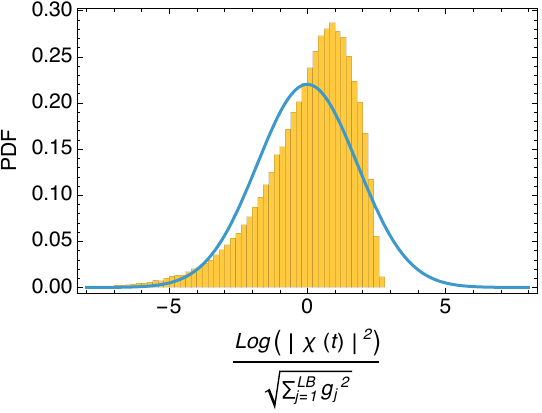}\includegraphics[width=4cm]{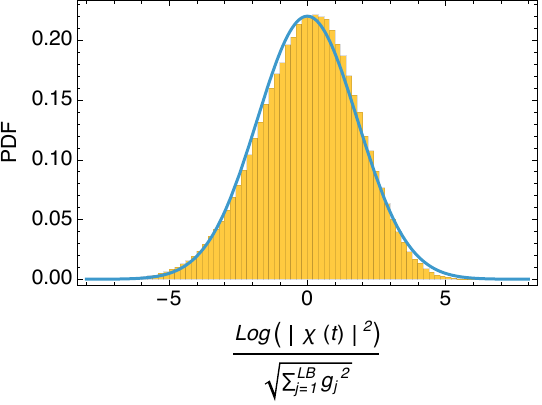}
\par\end{centering}
\begin{centering}
\includegraphics[width=4cm]{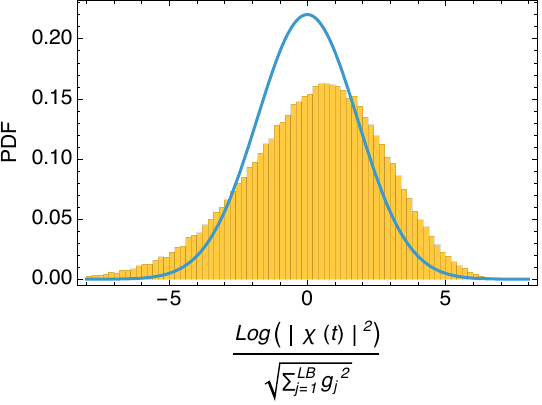}\includegraphics[width=4cm]{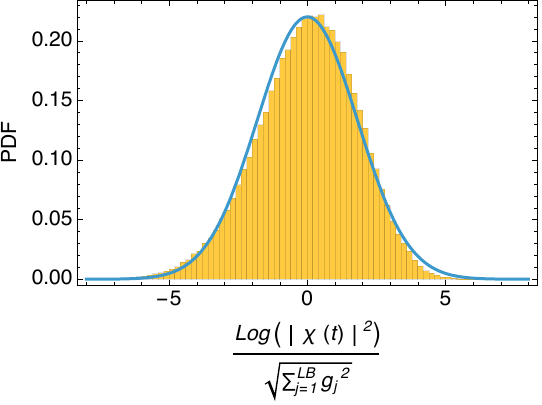}
\par\end{centering}
\caption{Check of Eq.~(\ref{eq:log-gaussian}) for the free XY model. Upper panels: $\left(h,\gamma\right)=\left(0.2,0.3\right)$
$L=8$ (left panel) and $L=170$ (right panel). Lower panels $\left(h,\gamma\right)=\left(0,0\right)$
$L=128$ (left panel) and $L=127$, prime (right panel). The histograms are obtained sampling $2\times 10^5$ time points uniformly in $[10^3,2\times 10^5]$. The continuous curve is the Gaussian $\mathsf{N}(0,\pi^2/3)$ Eq.~\ref{eq:log-gaussian}.
\label{fig:free-SFF}}
\end{figure}

\subsection{SFF distribution and moments in the SYK model}
\label{sub:SFF_SYK}
We now turn to a random (disordered) model and, in particular, we consider the generalized Sachdev-Ye-Kitaev (SYK-${k}$) model~\cite{sachdev1993gapless, KITP2015,Jasser_Odavic_Hamma_2025}.  The SYK-${k}$ model describes a ${k}$-body all-to-all interaction between $N_{\rm majo}$ Majorana Fermionic modes. The Hamiltonian reads
\begin{equation}
H^{{\rm SYK-}{k}} = (i)^{{k}/2} \sum\limits_{1 \le i_{1} < \hdots <  i_{{k}} \le N_{\rm majo}} J_{i_1 i_2 \hdots i_{{k}}} \chi_{i_1} \chi_{i_2} \hdots \chi_{i_{k}} \label{SYK}
\end{equation}
with ${k}$ being an even and positive integer number,  $\chi_i$  Majorana operators, where $N_{\rm majo} = 2 N$, with $N$ the number of qubit (spin) degrees of freedom hosting the Majorana modes. The couplings $J_{i_1, i_2, \dots, i_{k}}$ are identical, independent distributed $\rm (i.i.d.)$ Gaussian variables with mean and variance given by
\begin{equation}
\mathsf{E} [ J_{i_1, i_2, \dots, i_{k}} ]  = 0 \; , \; \; \; \mathsf{E} [J_{i_1, i_2, \dots, i_{k}}^2 ] = \frac{({k}-1)! J}{(N_{\rm majo})^{{k}-1}} \;. 
\label{variance}
\end{equation}
To perform exact diagonalization of the model and obtain its spectrum, we use the Jordan-Wigner transformation and map the Majorana operators into strings of Pauli operators~\cite{bravyi2002fermionic,sarosi2017ads}. For ${k} = 2$, the model constitutes a random quadratic model that is not chaotic~\cite{Liu_Chen_Balents_2018}, while for ${k} \ge 4$, the model is considered strongly chaotic. The spectrum of SYK-4 is non-degenerate when $ N_{\rm majo}\! \mod 8 = 0$ so that $ d_j = 1 \, \, \forall j=1,\ldots,D$, while it is entirely doubly degenerate, i.e.,~$d_j = 2, \, \, \forall j=1,\ldots,N_B$ when $N_{\rm majo}\! \mod 8 \neq 0$ due to a particle-hole symmetry~\cite{cotler_black_2017,CotlerErratum2018}. 
This implies that both Eqs.~(\ref{eq:generic_weights}) and (\ref{eq:Lyapunov_wiener}) are satisfied. The probability distribution of the eigenvalues is not known, but our numerical calculations and the computations in \cite{garcia-garcia_exact_2018} suggest that it is absolutely continuous with respect to the Lebesgue measure. This would imply that, besides the known degeneracies for $N_{\rm majo} \mod 8 \neq 0$ that we mentioned, with probability one, the spectrum is independent. This implies that, at infinite temperature $\rho=\1$, all the Theorems \ref{thm:CLT_lyapunovplus},   \ref{thm:Wiener}, and \ref{thm:moments}, apply to the SYK-4 model. At low enough temperature, the Lyapunov condition Eq.~(\ref{eq:generic_weights}) is not satisfied. According to Theorem~\ref{thm:CLT_lyapunovplus} the distribution of the random variable $Y_n$ is no longer Gaussian in the thermodynamic limit, and, correspondingly, the distribution of $|\chi(t)|^2/\overline{|\chi(t)|^2}$ is no longer exponential.

\begin{figure}
\begin{centering}
\includegraphics[width=8cm]{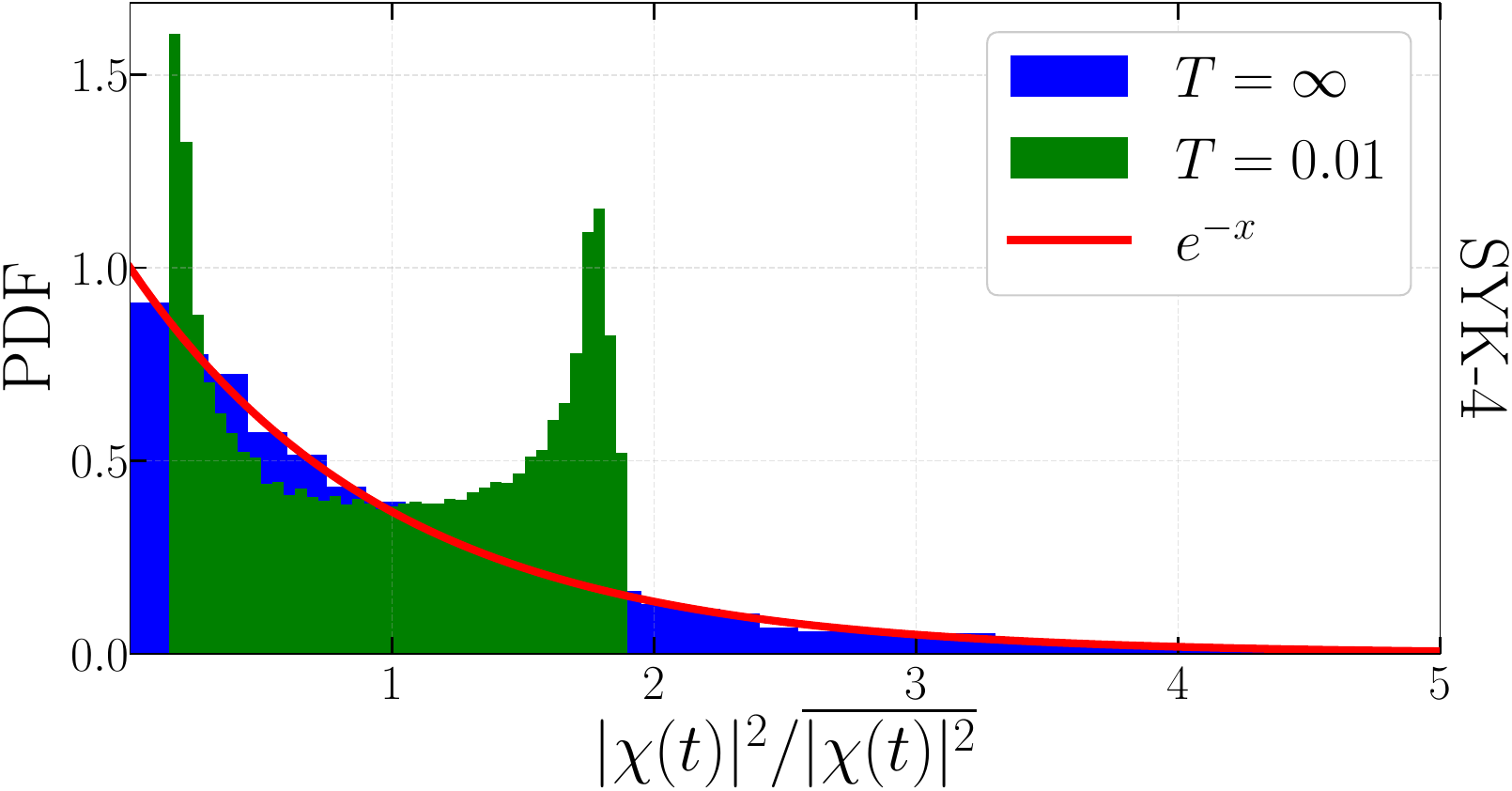}
\par\end{centering}
\caption{Probability density function (PDF) of the SFF at different temperatures for the SYK-4 model spectrum with $J = 1$ and $N_{\rm majo}=18$ Majorana modes.  Histograms of the SFF are obtained sampling $10^{5}$ time points uniformly distributed in  $[10^5, 2 \times 10^5]$ and  and for a single realization of the disorder couplings.}\label{fig:SYK4-1}
\end{figure}

In Fig.~\ref{fig:SYK4-1}, we show the full probability distribution of the normalized SFF both at infinite temperature and at very low temperature. It is apparent that in the $T=\infty$ case the distribution is $\mathsf{Exp}(1)$ which follows from Theorem~\ref{thm:CLT_lyapunovplus} while at $T=0.01$ the distribution is double-peaked signaling a breakdown of the CLT. 

In Table~\ref{tab:comparison} we show analogous results for the normalized moments of the SFF, $K_m$, defined as
\begin{equation}
K_{m}=\frac{\mathsf{E}_{x}\left[\overline{\left|\chi(t)\right|^{2m}}\right]}{\left(\mathsf{E}_{x}\left[\overline{\left|\chi(t)\right|^{2}}\right]\right)^{m}}=\frac{\mathsf{E}_{x}\left[I_{m}\right]}{\left(\mathsf{E}_{x}[I_{1}]\right)^{m}} \, ,
\end{equation}
where, to obtain better figures we also averaged over the random variables with $\mathsf{E}_{x}$. The ensemble average of formulae Eqs.~(\ref{eq:moments1}) and (\ref{eq:moments2}) give the correct moments while the Gaussian approximation $K_m \simeq m! $ fails at low temperature.

\begin{table}[t!]
\begin{tabular}{|c|c|c|c|c|}
\hline
\multirow{2}{*}{} & \multicolumn{2}{c|}{High Temperature $T = \infty$} & \multicolumn{2}{c|}{Low Temperature $T = 0.01$} \\
\cline{2-5}
                     & Gaussian ($\%$) & Exact ($\%$) & Gaussian ($\%$) & Exact ($\%$) \\
\hline
$K_2$ & $0.229 $ & $0.033 $ & $46.452$ & $0.011$  \\
\hline
$K_3$ & $1.018 $ & $0.434 $  & $79.754$& $0.007$ \\
\hline
$K_4$ & $1.484 $ & $0.316 $  & $93.936$& $0.025$\\
\hline
\end{tabular}\caption{Relative errors (in percentages) of the moments of the SFF of the SYK-4 model between the exact formulae (average of) Eqs.~(\ref{eq:recursion}) and the Gaussian approximation $K_m \simeq m!$. The numerical data are obtained by sampling $t$ uniformly in  $[10^5,2 \times 10^5]$, $10^{4}$ times and performing an additional ensemble average over 100 realizations for the SYK-4 model with $N_{\rm majo} = 18$ Majorana fermions. Increasing the number of realizations and time domain sampling leads to better agreement. The superior agreement of the exact expression is most evident for low temperatures.  }\label{tab:comparison}
\end{table}

\subsection{Evaluating the fractal dimension of the Brownian walker (or SFF) frontier}\label{sec:appfractral}

Intuitively, the fractal (or Hausdorff) dimension of a subset of the plane measures the roughness of the set. For example, for a smooth curve, the fractal dimension is $1$ while it is $2$ for a plane-filling curve \cite{falconer_geometry_1985}. 
One widely used approach for estimating the fractal dimension is the \textit{box-counting method} \cite{falconer_geometry_1985, Sokolovic_Mali_Odavic_Radosevic_Medvedeva_Botha_Shukrinov_Tekic_2017}. This technique involves overlaying a grid of square boxes of side length $\varepsilon$ over the domain containing the fractal and counting the number $N(\varepsilon)$ of boxes that intersect the fractal set. As $\varepsilon \to 0$, the fractal dimension $d_F$ is estimated by analyzing the scaling behavior of $N(\varepsilon)$, typically through a log-log regression
\[
d_F = \lim_{\varepsilon \to 0} \frac{\log N(\varepsilon)}{\log (1/\varepsilon)}.
\]
In practical implementations, this limit is approximated by evaluating $N(\varepsilon)$ at a range of decreasing box sizes and estimating the slope of the linear fit in a $\log$-$\log$ plot. 

We verified that our implementation accurately reproduces the exact known Hausdorff dimension for several fractals  with  dimensions in the range $0 \!< \! d_F \! < \! 2$. In general, there is only a definite window of scales over which a clear power-law appears. At finer scales the data have not sufficient resolution whereas at larger scales the power-law behavior did not set in yet. 
In our case, we found that such optimal window is sample dependent, making the estimation of the fractal dimension of the frontier more complex. 
To ensure statistical significance of our findings we preform the measurement of fractal dimension of the frontier on a ensemble of random walks (up to 20 different trajectories) for each of the considered regimes. Each member of the ensemble is generated by the same, exact, deterministic spectrum (in case of the disorder-free XXZ+NNN chain) but a different randomized value for the time parameter $t$ sampled uniformly at random in $t \in [1,2\times 10^5]$. 

Additionally, we note that isolating the fractal frontier in practice requires rasterizing numerical data into image format, a process that inherently reduces the precision of subsequent fractal dimension estimations. Extracting the frontier itself poses a significant technical challenge, as the random walker spans the continuous $\mathbb{R}^2$ plane. Standard convex hull algorithms, such as the gift wrapping algorithm~\cite{Leiserson_Rivest_Stein_2001}, are not well-suited for capturing rugged, scale-invariant structures. 
In this work we employed a combination of approaches using open-source image editing tools such as GIMP~\cite{GIMP}, along with image processing functionalities of the routine \textsf{ImageMeasurements} provided by the  \copyright Mathematica software. 

 In Fig.~\ref{fig:fractaldimension}, we illustrate the box-counting method for the  representative  pair of random-walks in Fig.~\ref{fig:mainplot} : one from the non-integrable (chaotic) case, and the other from the integrable XY chain with parameters $(h,\gamma) = (0.2, 0.3)$. These examples serve to visually demonstrate the qualitative differences between the two regimes. To ensure the reliability of our findings, the main text reports results averaged over multiple random-walks.

\begin{figure}[ht!]
\begin{centering}
\includegraphics[width=8cm]{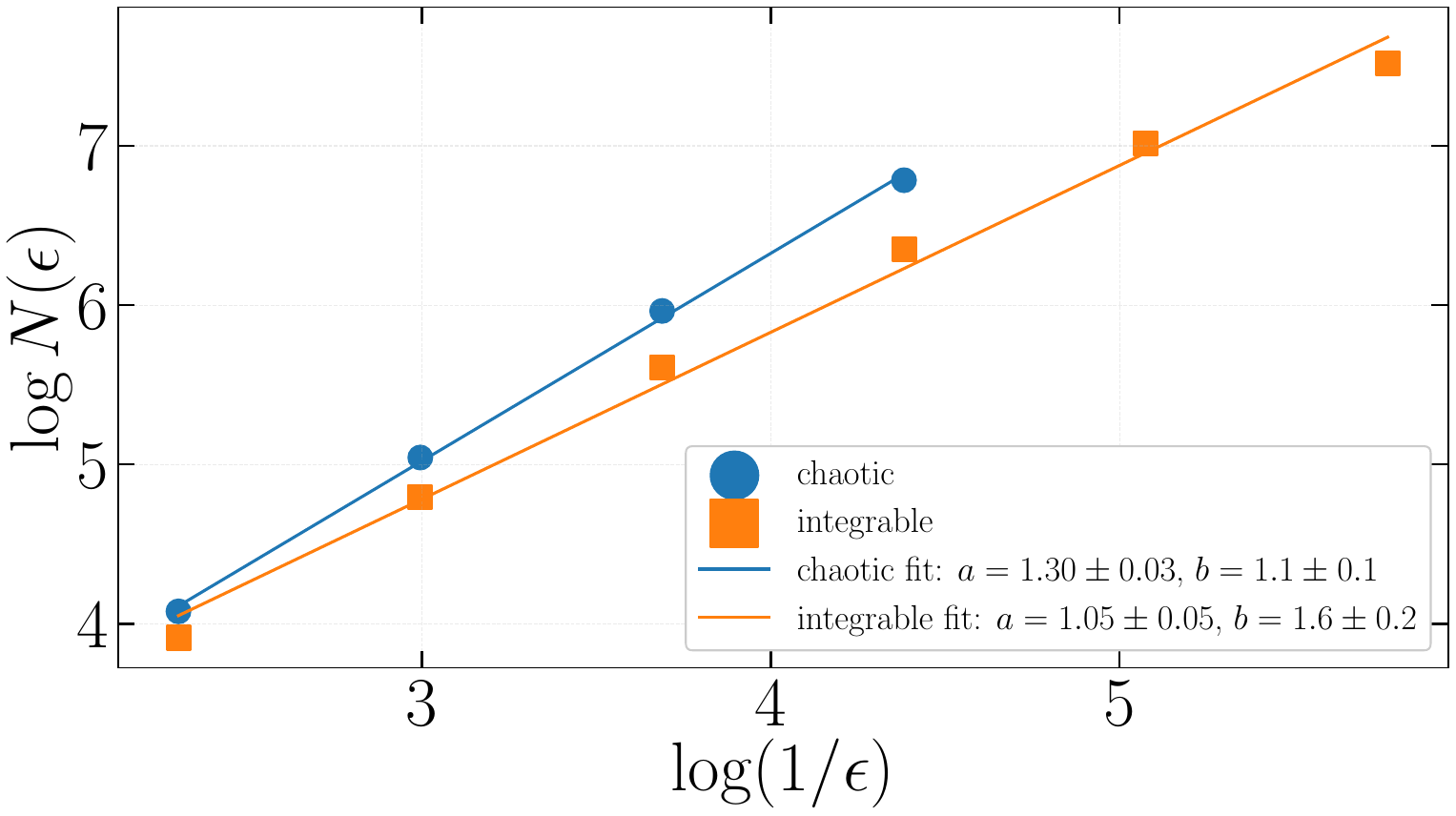}
\par\end{centering}
\caption{Box counting method and measurement of the fractal dimension for  representative fractals. Straight lines are fits to the linear function $f (x) = ax + b$. The data reflects the frontier fractal dimension of the data presented in Fig.~\ref{fig:mainplot}.  We consider $N = 14$ and the following parameters: integrable XY model spectrum with the Hamiltonian in Eq.~\eqref{eq:XY_model} with $(h,\gamma) = (0.2,0.3)$, chaotic spectrum for the Hamiltonian in Eq.~\eqref{XXZNNNHamiltonian} with $(\Delta, \alpha) = (0.5,0.4)$. The $\pm$ indicate the bounds related to the fitting errors. }\label{fig:fractaldimension}
\end{figure}

\clearpage

\bibliography{SFF_v5_refs}

\end{document}